\pdfoutput=1

\documentclass[preprint]{vgtc}               

\ifpdf
  \pdfoutput=1\relax                   
  \usepackage{graphicx}                
  \DeclareGraphicsExtensions{.pdf,.png,.jpg,.jpeg} 
\else
  \usepackage{graphicx}                
  \DeclareGraphicsExtensions{.eps}     
\fi%

\graphicspath{{figures/}{pictures/}{images/}{./}} 

\usepackage{microtype}                 
\PassOptionsToPackage{warn}{textcomp}  
\usepackage{textcomp}                  
\usepackage{mathptmx}                  
\usepackage{times}                     
\usepackage{cite}                      
\usepackage{tabu}                      
\usepackage{booktabs}                  

\usepackage[dvipsnames]{xcolor}

\onlineid{0}

\vgtccategory{Research}

\vgtcinsertpkg

\usepackage[utf8]{inputenc}
\usepackage{tikz}
\usepackage{amssymb}
\usepackage{amsmath}
\usepackage{mathtools}
\usepackage{todonotes}
\usepackage{subcaption}
\usepackage[linesnumbered,noend,ruled]{algorithm2e}
\usepackage{enumitem}
\usepackage{appendix}
\usepackage{amssymb}
\usepackage{pifont}
\usepackage{amsthm}
\usepackage{float}
\usepackage{doi}

\DeclareMathOperator{\troot}{root}

\DeclareMathOperator{\cost}{c}
\DeclareMathOperator{\pathstart}{start}
\DeclareMathOperator{\pathend}{end}
\DeclareMathOperator{\tpath}{path}
\DeclareMathOperator{\mappingedit}{edits}
\DeclareMathOperator{\mappingrelabel}{rel}
\DeclareMathOperator{\mappinginsert}{ins}
\DeclareMathOperator{\mappingdelete}{del}
\newcommand{\odedist}{\delta_{\text{1}}}
\newcommand{\edist}{\delta_{\text{E}}}

\newtheorem{theorem}{Theorem}
\newtheorem{lemma}{Lemma}
\newtheorem{definition}{Definition}

\title{A Deformation-based Edit Distance for Merge Trees}

\author{Florian Wetzels\thanks{e-mail: f\_wetzels13@cs.uni-kl.de}\\ %
        \scriptsize TU Kaiserslautern %
\and Christoph Garth\thanks{e-mail: garth@cs.uni-kl.de}\\ %
     \scriptsize TU Kaiserslautern }

\teaser{
 \centering 
 \resizebox{0.9\linewidth}{!}{
 \begin{tikzpicture}[yscale=0.8]
 \definecolor{lightred}{rgb}{1,0.6,0.6}
 \draw[help lines, color=gray!30, dashed] (-3,0) grid (21,11);
 
 \node[draw,circle,fill=gray] at (0, 1) (root) {A};
 \node[draw,circle,fill=gray] at (0, 4) (s1) {B};
 \node[draw,circle,fill=gray] at (-1, 6) (s2) {C};
 
 \node[draw,circle,fill=red] at (-2, 10) (m1) {D};
 \node[draw,circle,fill=red] at (0, 8) (m2) {E};
 \node[draw,circle,fill=red] at (2, 9) (m3) {F};
 
 \draw[gray,very thick] (root) -- (s1);
 \draw[gray,very thick] (s1) -- (s2);
 \draw[gray,very thick] (s2) -- (m1);
 \draw[gray,very thick] (s2) -- (m2);
 \draw[gray,very thick] (s1) -- (m3);
 
 \node[draw,circle,fill=gray] at (6+0, 1) (root') {A};
 \node[draw,circle,fill=gray] at (6+0, 4) (s1') {B};
 
 \node[draw,circle,fill=red] at (6-2, 10) (m1') {D};
 \node[draw,circle,fill=red] at (6+2, 9) (m3') {F};
 
 \draw[gray,very thick] (root') -- (s1');
 \draw[gray,very thick] (s1') -- (m1');
 \draw[gray,very thick] (s1') -- (m3');
 
 \node[draw,circle,fill=gray] at (12+0, 1) (root'') {A};
 \node[draw,circle,fill=gray] at (12+0, 3) (s1'') {B};
 
 \node[draw,circle,fill=red] at (12-2, 9) (m1'') {D};
 \node[draw,circle,fill=red] at (12+2, 8) (m3'') {F};
 
 \draw[gray,very thick] (root'') -- (s1'');
 \draw[gray,very thick] (s1'') -- (m1'');
 \draw[gray,very thick] (s1'') -- (m3'');
 
 \node[draw,circle,fill=gray] at (18+0, 1) (root'') {A};
 \node[draw,circle,fill=gray] at (18+0, 3) (s1'') {B};
 \node[draw,circle,fill=gray] at (18+1.2, 6) (s2'') {C};
 
 \node[draw,circle,fill=red] at (18-2, 9) (m1'') {D};
 \node[draw,circle,fill=red] at (18+0.2, 7.5) (m2'') {E};
 \node[draw,circle,fill=red] at (18+2, 8) (m3'') {F};
 
 \draw[gray,very thick] (root'') -- (s1'');
 \draw[gray,very thick] (s1'') -- (s2'');
 \draw[gray,very thick] (s1'') -- (m1'');
 \draw[gray,very thick] (s2'') -- (m2'');
 \draw[gray,very thick] (s2'') -- (m3'');
 
 \draw[->,ultra thick] (1.5,4.42) -- (4.5,4.42) node [midway,above] {\large Delete (E,C)};
 \draw[->,ultra thick] (1.5,4.42) -- (4.5,4.42) node [midway,below] {\large $\text{cost} = 2$};
 
 \draw[->,ultra thick] (7.5,4.42) -- (10.5,4.42) node [midway,above] {\large Relabel (B,A)};
 \draw[->,ultra thick] (7.5,4.42) -- (10.5,4.42) node [midway,below] {\large $\text{cost} = 1$};
 
 \draw[->,ultra thick] (13.5,4.42) -- (16.5,4.42) node [midway,above] {\large Insert (E,C)};
 \draw[->,ultra thick] (13.5,4.42) -- (16.5,4.42) node [midway,below] {\large $\text{cost} = 1.5$};

 \end{tikzpicture}
 }
 \caption{Illustration of the edit operations on abstract merge trees. For all three types, the cost function is based on persistence change, i.e.\ change in y-range of the modified edge (see coordinate grid).}
 \label{fig:edit_operations_abstract}
}

\abstract{In scientific visualization, scalar fields are often compared through edit distances between their merge trees. Typical tasks include ensemble analysis, feature tracking and symmetry or periodicity detection. Tree edit distances represent how one tree can be transformed into another through a sequence of simple edit operations: relabeling, insertion and deletion of nodes. In this paper, we present a new set of edit operations working directly on the merge tree as an geometrical or topological object: the represented operations are deformation retractions and inverse transformations on merge trees, which stands in contrast to other methods working on branch decomposition trees. We present a quartic time algorithm for the new edit distance, which is branch decomposition-independent and a metric on the set of all merge trees.%
} 


\begin{document}


\firstsection{Introduction}
\label{section:introduction}

\maketitle

Measuring similarity or dissimilarity of scalar fields, as well as finding similar features or mappings between them, is an important tool in scientific visualization for the analysis of ensemble data or time series~\cite{DBLP:journals/cgf/YanMSRNHW21},
specifically for tasks such as feature tracking, clustering, and the detection of periodicity or self-similarity. Both problems can be and have been addressed through the use of edit distances on merge trees~\cite{DBLP:journals/tvcg/PontVDT22,DBLP:journals/tvcg/SridharamurthyM20,DBLP:journals/cgf/LohfinkWLWG20,wetzels2022branch}, an abstract representation of the topology of sub-level sets or super-level sets of scalar functions. Tree edit distances, which come in a large variety of different forms~\cite{treeEditSurvey}, are well-suited for these tasks since, typically (i.e.\ for most variants), they are efficiently to compute, induce mappings between the edges of the trees (which correspond to topological features), fulfill the metric properties, are very intuitive to understand and have great flexibility through the use of different base metrics on the labels of the trees. Furthermore, working on topological abstractions such as merge trees has a huge impact on performance, as these structures usually stay rather small in comparison to the actual data domain.

The operations in classic tree edit distances are node-insertion, node-deletion and node-relabel. For merge trees, this set of edit operations is not coherent with the intuitive way to transform them, since merge trees are actually continuous objects, whereas node-labeled trees are not. 
Formally, this means that applying the classic edit operations to a merge tree may not result in another merge tree, or an invalid one.
Figure~\ref{fig:classicVsContinuousEdit} illustrates this in more detail. Previous approaches~\cite{DBLP:journals/tvcg/PontVDT22,DBLP:journals/tvcg/SridharamurthyM20,DBLP:journals/cgf/SaikiaSW14} overcame this issue by working on branch decomposition trees (BDTs) of merge trees rather than working on the merge tree itself. However, this comes with the downside of using fixed branch decompositions, which are very susceptible to small-scale perturbations in the data. Although this problem has also been overcome recently~\cite{wetzels2022branch} through the new concept of branch mappings, the solution came at the cost of losing the metric property.

Furthermore, both approaches (using fixed BDTs or branch mappings) lose the direct connection to modifying operations on merge trees. While some connection is still there (e.g.\ all operations on BDTs can be interpreted as non-primitive operations on merge trees), which will be discussed in Section~\ref{section:discussion}, the corresponding operations differ significantly from the classic model on arbitrary trees, and are not as intuitive. In particular, the branch based edit distances focus more on the induced mappings than the actual edit operations.

This paper introduces an edit distance for merge trees based on a new set of edit operations which are specifically tailored to deformations of merge trees, with a cost function based on the typical drawings or embeddings of merge trees, making it highly intuitive.

\begin{figure}[]
    \centering
    \resizebox{\linewidth}{!}{
    \begin{tikzpicture}[yscale=0.5]
    
    \node[draw,circle,fill=gray!100,minimum width=0.7cm] at (0, 0) (root_1) {};
    \node[draw,circle,fill=gray!100,minimum width=0.7cm] at (0, 3) (s1_1) {};
    \node[draw,circle,fill=red!80,minimum width=0.7cm] at (-2, 12) (m1_1) {};
    \node[draw,circle,fill=red!80,minimum width=0.7cm] at (2, 12) (m2_1) {};
    \node[draw,circle,fill=gray!100,minimum width=0.7cm] at (1.333, 9) (s2_1) {};
    \node[draw,circle,fill=red!80,minimum width=0.7cm] at (0.3, 11) (m3_1) {};
    \draw[gray,very thick] (root_1) -- (s1_1);
    \draw[gray,very thick] (s1_1) -- (m1_1);
    \draw[gray,very thick] (s1_1) -- (s2_1);
    \draw[gray,very thick] (s2_1) -- (m2_1);
    \draw[gray,very thick] (s2_1) -- (m3_1);
    
    \node[draw,circle,fill=gray!100,minimum width=0.7cm] at (0+8, 0) (root_2) {};
    \node[draw,circle,fill=gray!100,minimum width=0.7cm] at (0+8, 3) (s1_2) {};
    \node[draw,circle,fill=red!80,minimum width=0.7cm] at (-2+8, 12) (m1_2) {};
    \node[draw,circle,fill=red!80,minimum width=0.7cm] at (2+8, 12) (m2_2) {};
    \draw[gray,very thick] (root_2) -- (s1_2);
    \draw[gray,very thick] (s1_2) -- (m1_2);
    \draw[gray,very thick] (s1_2) -- (m2_2);
    
    \node[draw,circle,fill=gray!100,minimum width=0.7cm] at (0-8, 0) (root_3) {};
    \node[draw,circle,fill=gray!100,minimum width=0.7cm] at (0-8, 3) (s1_3) {};
    \node[draw,circle,fill=red!80,minimum width=0.7cm] at (-2-8, 12) (m1_3) {};
    \node[draw,circle,fill=red!80,minimum width=0.7cm] at (2-8, 12) (m2_3) {};
    \node[draw,circle,fill=gray!100,minimum width=0.7cm] at (1.333-8, 9) (s2_3) {};
    \draw[gray,very thick] (root_3) -- (s1_3);
    \draw[gray,very thick] (s1_3) -- (m1_3);
    \draw[gray,very thick] (s1_3) -- (s2_3);
    \draw[gray,very thick] (s2_3) -- (m2_3);
    
    \draw[->,ultra thick] (-1.5,6) to[bend left=20] (-6.5,6) node [midway,above] {};
    \draw[->,ultra thick] (1.5,6) to[bend right=20] (6.5,6) node [midway,above] {};
    \node[] at (-4, 6.7) (label_1) {\huge classic delete};
    \node[] at (4, 6.7) (label_2) {\huge continuous delete};
    
    \end{tikzpicture}
    }
    \caption{Illustration of the difference between the classic discrete edit operations and the new continuous edit operations, exemplary for a deletion. If we apply a classic node-deletion or edge-contraction to the middle tree, we get the tree on the left. The result is not a merge tree. The desired result is shown in the right tree. Note that the problem cannot be fixed through deleting the remaining node with another classic edit operation, because a deletion always removes an edge as well.}
    \label{fig:classicVsContinuousEdit}
\end{figure}
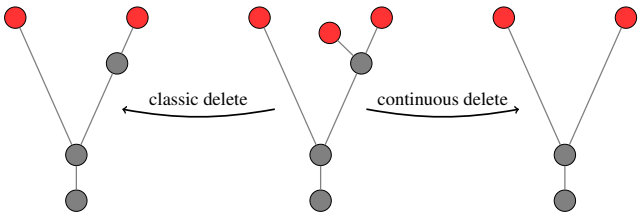

\subsection*{Contribution}

In particular, we present an edit distance between merge trees with its corresponding mappings, which we call \emph{path mappings}, that
\begin{itemize}
    \item corresponds to an optimal sequence of deformations,
    \item is independent of a fixed branch decomposition,
    \item fulfills the metric properties,
    \item shows better practical performance than branch mappings.
\end{itemize}
Furthermore, we provide an open-source implementation of our distance and showcase its utility by replicating previous results for clustering and periodicity detection.

\smallskip
In the remainder of this section, we cover related work. Section~\ref{section:preliminaries} recaps important definitions and concepts. In Section~\ref{section:pathmappings}, we introduce the new edit distance and path mappings, and show their core properties needed for implementation and application. Section~\ref{section:experiments} discusses the actual algorithms and the results of our experiments. In Section~\ref{section:discussion}, we discuss the choice of edit operations and the relation to previous methods. We conclude the paper with an outlook to future work in Section~\ref{section:conclusion}.

\subsection*{Related Work}

Topological descriptors and abstractions play a key role for the analysis of scalar fields in scientific visualization. Many of them have been used for the task of scalar field comparison: merge trees or their generalizations contour trees and reeb graphs, as well as others like extremum graphs or persistence diagrams. A recent survey of this topic can be found in~\cite{DBLP:journals/cgf/YanMSRNHW21}. A survey on topology based visualization methods in general is given by Heine et al.\ in~\cite{DBLP:journals/cgf/HeineLHIFSHG16}. Tree edit distances in general have been introduced by Tai for ordered rooted trees~\cite{DBLP:journals/jacm/Tai79} and adapted for unordered rooted trees by Zhang~\cite{DBLP:journals/ipl/ZhangSS92}. In this paper, we use the one-degree edit distance, which has been introduced by Selkow~\cite{DBLP:journals/ipl/Selkow77} for ordered rooted trees and which is a special variant of the constrained edit distance~\cite{DBLP:journals/algorithmica/Zhang96}. A survey on the various versions of tree edit distances and related problems can be found in~\cite{treeEditSurvey}. We now quickly review those methods from the field of scientific visualization that are closest to the here presented distance measure.

The constrained edit distance on merge trees (actually working on a fixed BDT) has been applied by Sridharamurthy et al.~\cite{DBLP:journals/tvcg/SridharamurthyM20} to various visualization tasks including periodicity detection or clustering. They adapted their method for the use in self-similarity detection in~\cite{DBLP:journals/corr/abs-2111-04382}. Saikia et al.\ addressed the same task through the use of the one-degree edit distance on ordered BDTs~\cite{DBLP:journals/cgf/SaikiaSW14}. Pont et al.\ used the one-degree edit distance on unordered BDTs to define a Wasserstein distance for merge trees and compute geodesics and barycenters on them~\cite{DBLP:journals/tvcg/PontVDT22}. Lohfink et al.\ used tree alignments, another special form of tree edit distances (see~\cite{DBLP:conf/cpm/JiangWZ94,treeEditSurvey}), to compute a supertree of all contour trees in an ensemble~\cite{DBLP:journals/cgf/LohfinkWLWG20,DBLP:conf/visualization/LohfinkGWVG21}. Wetzels et al.\ defined the concept of branch mappings~\cite{wetzels2022branch}, an edit distances that works on BDTs but independent of a fixed branch decomposition by optimizing the considered branch decomposition on the fly. Other distances on merge trees and contour trees that are not specifically edit distances can be found in~\cite{morozov2013interleaving,DBLP:books/daglib/p/BeketayevYMWH14,DBLP:journals/tvcg/ThomasN11,DBLP:journals/tvcg/YanWMGW20}.

Distance measures on other topological graph structures have been applied as well. For the reeb graph, of which contour trees are just a special case, an edit distance was proposed in~\cite{reebgrapheditdistance}. Actually, the here presented edit distance can be seen as an adaptation of the reeb graph edit distance to merge trees (note however that tree edit distances are more than just a special case of graph edit distances, as the allowed operations differ significantly and the graph edit distance has a much higher complexity). Other examples for distances between reeb graphs can be found in~\cite{reebgraphdistance,localEquivalence,categorifiedreebgraphs}. In~\cite{DBLP:conf/apvis/NarayananTN15}, a distance measure for extremum graphs was introduced.

\section{Preliminaries}
\label{section:preliminaries}

In this section, we recap the core definitions from computational topology and graph theory that are needed to define the edit distance for merge trees and to study its properties. For basic notions on topological spaces and simplicial complexes, we refer to~\cite{DBLP:books/daglib/0025666}.

\subsection*{Merge Trees}
\label{section:merge_trees}

Given a $d$-manifold $\mathbb{X}$ with a continuous map $f: \mathbb{X} \rightarrow \mathbb{R}$, its \emph{Join Tree} and \emph{Split Tree} represent the connectivity of its sub-level sets and super-level sets.

The join tree of $\mathbb{X},f$ is the quotient space $\mathbb{X}/{\sim}$ under the equivalence relation $\sim$, where $x \sim y$ if $f(x) = f(y)$ and $x$ and $y$ belong to the same connected component of the sublevel set $f^{-1}((-\infty,f(x)])$. A split tree of $\mathbb{X},f$ is defined in the same way by just replacing  the sublevel set $f^{-1}((-\infty,f(x)])$ with the superlevel set $f^{-1}([f(x),-\infty))$.

We use the terms \emph{critical point}, \emph{maximum}, \emph{minimum}, \emph{saddle} and \emph{path} (in a topological space, not in discrete graphs) following the definitions in~\cite{carr2004topological}.

A merge tree is, in essence, a 1-dimensional simplicial complex. I.e.\ for each merge tree $\mathbb{X}/{\sim},f$ there is a simplicial complex $K$ of dimension 1 with a scalar function $f'$ such that its underlying space $|K|$ is not only homeomorphic to $\mathbb{X}/{\sim}$, but there is a homeomorphism $h : \mathbb{X}/{\sim} \rightarrow |K|$ that preserves the scalar function, $f(x) = f'(h(x))$ for all $x \in \mathbb{X}/{\sim}$, and with $\text{Vert}(K)$ being exactly the critical points of $\mathbb{X},f$. We call this simplicial complex $\mathcal{T}(\mathbb{X},f)$, and we will refer to the underlying spaces $|\mathcal{T}(\mathbb{X},f)|$ of these structures as merge trees. This means that by merge trees we do not only denote quotient spaces of scalar fields but actually all spaces that are homeomorphic to them through a scalar function-preserving homeomorphism.

As a next step, we define an abstract model for merge trees and introduce basic notation for graphs and trees used in this paper. After that, we will discuss the relation of continuous merge trees (as quotient spaces) and abstract merge trees.

\subsection*{Abstract Merge Trees}
\label{section:abstract_merge_trees}

Throughout this paper, we will consider rooted trees as directed graphs with parent edges. I.e.\ a rooted tree $T$ is a directed graph with vertex set $V(T)$, edge set $E(T) \subseteq V(T) \times V(T)$ and a unique root, denoted $\troot(T)$. We call a node $c \in V(T)$ a child of node $p \in V(T)$, if $(c,p) \in E(T)$, and, conversely, $p$ the parent of $c$. For a node $p$, we denote its number of children by $\deg_T(p) \coloneqq |\{ c \mid (c,p) \in E(T) \}|$. Furthermore, we denote the empty tree, consisting only of a single node and no edges, by $\bot$.

Those rooted trees that can be interpreted as merge trees for some domain of dimension at least~$2$ will be called \emph{abstract} merge trees. They are the center objects of this paper.

\begin{definition}
An unordered, rooted tree $T$ of (in general) arbitrary degree with edge labels $\ell:V(T) \rightarrow \mathbb{R}_{>0}$ is an \emph{Abstract Merge Tree} if the following properties hold:
\begin{itemize}
    \item The root node has degree one, $\deg_T( \troot (T) ) = 1$
    \item All inner nodes have a degree of at least two,\\ $\deg_T(v) \neq 1$ for all $v \in V(T)$ with $v \neq \troot (T)$
\end{itemize}
\end{definition}

Note that we do not use node labels representing the scalar value of the original critical points but rather edge labels directly representing the persistence of edges (we use the term persistence for the length of edges and paths, since this similarity measure is just an adaptation of the persistence of branches and to distinguish it from the length of paths, i.e.\ number of edges, in abstract trees). We chose this for two reasons: first, persistence of edges or paths are the properties of interest anyway and second, this gives a lot of flexibility in terms of representing similar properties like edge volume or region size via the same abstract definition. Since we are usually not interested in shifts of the whole merge tree to higher scalar values, this should not influence the practicality of the definition. Throughout this paper, we will often just write $T$ instead of $T,\ell$, but if we do, it should be clear from the context that the abstract merge tree has a label function attached.

Since the root of an abstract merge tree always has degree one and inner nodes do not, subtrees rooted in an inner node are not abstract merge trees themselves.
Therefore, we identify subtrees by root edges, rather than root nodes:
Formally, for a node $p \in V(T)$ with child $c \in V(T)$, the subtree rooted in $(c,p) \in E(T)$, denoted by $T[(c,p)]$ is defined by the vertex set
$$V(T[(c,p)])=\{ c,p \} \cup \{ v \mid c \text{ is an ancestor of } v \}$$
and the induced edge set.
Given an abstract merge tree $T$ with subtree $T'$ rooted in the edge $(c,p)$, we define $T-T'$ to be the tree $T''$, which we obtain by removing all edges and all vertices of $T'$ from $T$ except the root $p$.
If $\deg_T(p) = 2$,
then we also remove it from $T''$, as otherwise $p$ would be an inner node of degree one in $T''$. With this definition, it holds that
$T'$ and $T''$ are abstract merge trees as well.

As for general graphs, a \emph{path} of length $k$ in an abstract merge tree $T$ is a sequence of vertices $p=v_1 ... v_k \in V(T)^k$ with $(v_{i},v_{i-1}) \in E(T)$ for all $2 \leq i \leq k$ and $v_i \neq v_j$ for all $1 \leq i,j \leq k$. Note the strict root-to-leaf direction: we only consider monotone paths.
For a path $p=v_1 ... v_k$, we denote its first vertex by $\pathstart(p) \coloneqq v_1$, its last vertex by $\pathend(p) \coloneqq v_k$ and the set of all paths of a tree $T$ by $\mathcal{P}(T)$.

In an abstract merge tree $T$, each node $v \in V(T)$ has a unique path connecting it to the root of the tree $\troot(T)$. Each node $u \neq v$ on this path is called an \emph{ancestor} of $v$ and we denote the unique path connecting $u$ and $v$ by $\tpath_T(v,u)$.

We lift the label function $\ell$ of an abstract merge tree $T$ from edges to paths in the following way: $\ell(v_1...v_k) = \sum_{2 \leq i \leq k} \ell((v_i,v_{i-1}))$.

\subsection*{Continuous Merge Trees to Abstract Merge Trees}

Now, we show how the two introduced concepts of merge trees relate. As stated above, abstract merge trees represent those trees that can be interpreted as merge tree for some domain. Formally, we denote the continuous merge tree of a scalar field $\mathbb{X},f$ by $\mathcal{T}(\mathbb{X},f)$ and the corresponding abstract merge tree by $T(\mathbb{X},f)$, which we define in the following. Again, we closely stick to the definitions for contour trees in~\cite{carr2004topological}.

The vertex set $V(T(\mathbb{X},f))$ is the set of critical points in $|\mathcal{T}(\mathbb{X},f)|$. For the edges, we have $(u,v) \in E(T(\mathbb{X},f))$ if and only if $f(v)<f(u)$ and there is an $f$-monotone path connecting $u$ and $v$ such that $u$ and $v$ are the only critical points on this path. We define the label $\ell_f((u,v))$ to be $|f(u)-f(v)|$. Note the strong correspondence between an actual merge tree $|\mathcal{T}|$ and the abstract merge tree $T$: the vertices in $V(T)$ are exactly the vertices of the simplicial complex $\mathcal{T}$ and the edges $E(T)$ are exactly the 1-simplices in $\mathcal{T}$. If only a merge tree $|\mathcal{T}|$ is given without the original domain, we can therefore denote by $T(\mathcal{T})$ its corresponding abstract merge tree. Furthermore, for each abstract merge tree $T$, we can define a simplicial complex $\mathcal{T}(T)$ such that $|\mathcal{T}(T)|$ is a merge tree and $T(\mathcal{T}(T)) = T$ by using an arbitrary embedding/drawing of the abstract merge tree.

\section{Deformation-based Edit Distance}

In this section, we introduce the new edit distance for merge trees. First, we define two models of edit operations: an intuitive description for a set of deformations on continuous merge trees, which is mainly used to guide the motivation of the new distance, as well as a formal definition of edit operations on abstract merge trees. Next, we study the relationship of the two models before presenting the underlying algorithmic concept of path mappings and their recursive structure. We postpone an in-depth discussion of the choice of edit operations to Section~\ref{section:discussion}.

\subsection{Edit Operations}

We begin by defining two sets of edit operations on merge trees, one of which works on the continuous objects, i.e.\ merge trees as the quotient spaces of the original domains, and the other one works on discrete abstract merge trees and is used for computation.

\subsubsection*{Continuous Edit Operations}
For continuous trees, the goal is to define an intuitive set of edit operations, which should be able to transform any two merge trees into each other and operate only locally on edges and nodes. Similar to deletions and insertions on arbitrary trees, we need edit operations changing the size of the tree. Since the labels in merge trees are the lengths of edges, relabel operations fall into the same category. Therefore, we identified only two types of operations:
\begin{itemize}
    \item Shrinking an arc, possibly deleting it.
    \item Extending an arc or inserting a new one.
\end{itemize}
These operations have the following intuition: they strongly resemble deformation retractions on merge trees: since a deformation retraction on a merge tree always yields another merge tree (or a single point, which is formally also a tree), all they can do is shrink edges, or possibly remove them. The corresponding inverse transformations are extending or inserting arcs. Hence, a sequence of shrinking operations transforms $\mathcal{T}_1$ into $\mathcal{T}_2$ if and only if $\mathcal{T}_2$ is a deformation retract of $\mathcal{T}_1$, and conversely, a sequence of extending operations transforms $\mathcal{T}_1$ into $\mathcal{T}_2$ if and only if $\mathcal{T}_1$ is a deformation retract of $\mathcal{T}_2$. Next, we define the equivalent model on abstract merges and then prove the correspondence to deformation retractions on this model.

\begin{figure*}[t]
 
 \centering
 \resizebox{0.9\linewidth}{!}{
 \begin{tikzpicture}[yscale=0.65]
 \definecolor{lightred}{rgb}{1,0.6,0.6}
 
 
 
 \node[draw,circle,fill=gray!70] at (0, 1) (root) {\textbf{\Large A}};
 \node[draw,circle,fill=gray!70] at (0, 4) (s1) {\textbf{\Large B}};
 \node[draw,circle,fill=gray!70] at (-1, 7) (s2) {\textbf{\Large C}};
 
 \node[draw,circle,fill=red!70] at (-2, 10) (m1) {\textbf{\Large D}};
 \node[draw,circle,fill=red!70] at (0, 8.5) (m2) {\textbf{\Large E}};
 \node[draw,circle,fill=red!70] at (2, 9) (m3) {\textbf{\Large F}};
 
 \draw[green,line width=3pt] (root) -- (s1);
 \draw[purple,line width=3pt] (s1) -- (s2);
 \draw[blue,line width=3pt] (s2) -- (m1);
 \draw[cyan,line width=3pt] (s2) -- (m2);
 \draw[orange,line width=3pt] (s1) -- (m3);
 
 
 \node[draw,circle,fill=gray!70] at (8+0, 1) (root) {\textbf{\Large A}};
 \node[draw,circle,fill=gray!70] at (8+0, 4) (s1) {\textbf{\Large B}};
 \node[draw,circle,fill=gray!70] at (8-1, 7) (s2) {\textbf{\Large C}};
 
 \node[draw,circle,fill=red!70] at (8-2, 10) (m1) {\textbf{\Large D}};
 \node[draw,circle,fill=red!70] at (8+0, 8.5) (m2) {\textbf{\Large E}};
 
 \draw[green,line width=3pt] (root) -- (s1);
 \draw[purple,line width=3pt] (s1) -- (s2);
 \draw[blue,line width=3pt] (s2) -- (m1);
 \draw[cyan,line width=3pt] (s2) -- (m2);
 
 
 \node[draw,circle,fill=gray!70] at (16+0, 1) (root') {\textbf{\Large A}};
 \node[draw,circle,fill=gray!70] at (16+0, 7) (s1') {\textbf{\Large C}};
 
 \node[draw,circle,fill=red!70] at (16-2, 10) (m1') {\textbf{\Large D}};
 \node[draw,circle,fill=red!70] at (16+2, 8.5) (m2') {\textbf{\Large E}};
 
 \draw[green,line width=3pt] (root') -- (16,4);
 \draw[purple,line width=3pt] (16,4) -- (s1');
 \draw[blue,line width=3pt] (s1') -- (m1');
 \draw[cyan,line width=3pt] (s1') -- (m2');
 
 
 \node[draw,circle,fill=gray!70] at (24+0, 1) (root') {\textbf{\Large A}};
 \node[draw,circle,fill=gray!70] at (24+0, 7) (s1') {\textbf{\Large C}};
 
 \node[draw,circle,fill=red!70] at (24-2, 10) (m1') {\textbf{\Large D}};
 
 \draw[green,line width=3pt] (root') -- (24,4);
 \draw[purple,line width=3pt] (24,4) -- (s1');
 \draw[blue,line width=3pt] (s1') -- (m1');
 
 
 \node[draw,circle,fill=gray!70] at (32+0, 1) (root') {\textbf{\Large A}};
 
 \node[draw,circle,fill=red!70] at (32+0, 10) (m1') {\textbf{\Large D}};
 
 \draw[green,line width=3pt] (root') -- (32,4);
 \draw[purple,line width=3pt] (32,4) -- (32,7);
 \draw[blue,line width=3pt] (32,7) -- (m1');
 
 \draw[->,line width=2.5pt] (0.5,10.5) to[bend left=30] node [midway,below=3pt] {\huge Delete (F,B)} (15.5,10.5);
 \draw[->,line width=2.5pt] (15.5,0) to[bend left=30] node [midway,above=3pt] {\huge Insert (F,B)} (0.5,0);
 
 \draw[->,line width=2.5pt] (16.5,10.5) to[bend left=30] node [midway,below=3pt] {\huge Delete (E,C)} (31.5,10.5);
 \draw[->,line width=2.5pt] (31.5,0) to[bend left=30] node [midway,above=3pt] {\huge Insert (E,C)} (16.5,0);
 
 \draw[->,dashed,line width=2.5pt] (1.5,6) -- (6.5,6);
 \draw[->,dashed,line width=2.5pt] (6.5,3) -- (1.5,3);
 \draw[->,dashed,line width=2.5pt] (9.5,6) -- (14.5,6);
 \draw[->,dashed,line width=2.5pt] (14.5,3) -- (9.5,3);
 
 \draw[->,dashed,line width=2.5pt] (17.5,6) -- (22.5,6);
 \draw[->,dashed,line width=2.5pt] (22.5,3) -- (17.5,3);
 \draw[->,dashed,line width=2.5pt] (25.5,6) -- (30.5,6);
 \draw[->,dashed,line width=2.5pt] (30.5,3) -- (25.5,3);

 \end{tikzpicture}
 }
 
 \caption{Deriving a path mapping from edit operations: If we delete the edge $(F,B)$ in the leftmost tree, this happens in two steps. First, we remove the edge and the leaf node $F$, then we merge the edges $(C,B)$ and $(B,A)$, since $B$ has degree one. The merging step induces a clear correspondence between the path $ABC$ and the edge $(C,A)$. This correspondence is transitive for multiple edit operations, hence, we can map the path $ABCD$ to the edge $(D,A)$ after deleting $(E,C)$. The mappings for insert operations can be derived in the same way.}
\label{fig:pathMappingIllustration}
 
\end{figure*}
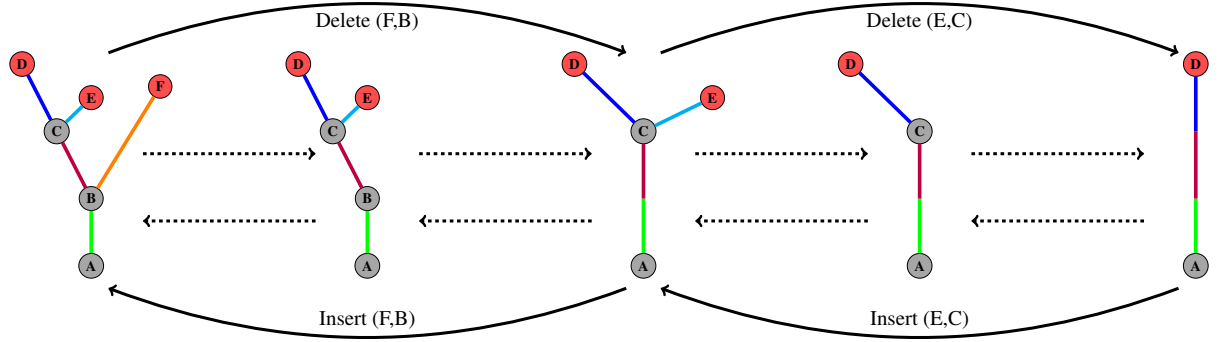

\subsubsection*{Abstract Edit Operations}
The edit distance for \emph{abstract} merge trees is now defined in a more formal way: we consider the following edit operations that transform an abstract merge tree $T,\ell$ into another abstract merge tree $T',\ell'$:
\begin{itemize}
    \item Edge relabel: changing the length of an edge $(c,p)$ to a new value $v \in \mathbb{R}_{>0}$, i.e.\ $T'=T$, $\ell'((c,p)) = v$ and $\ell'(e)=\ell(e)$ for all $e \neq (c,p)$.
    \item Edge contraction: remove an edge from the tree and merge the two nodes. Then, remove the parent node if it had only two children originally. Formally, for a node $p$ with children $c_0...c_k$ and parent $p'$, we define $T'$ after contracting $(c_i,p)$ as follows: if $k>1$, we have
    $$V(T') = V(T) \setminus \{c_i\}, E(T') = E(T) \setminus \{(c_i,p)\},$$
    and otherwise, if $k=1$, we have
    $$V(T') = V(T) \setminus \{c_i,p\},$$
    $$E(T') = (E(T) \cup \{(c_{1-i},p')\}) \setminus \{(c_i,p),(c_{1-i},p),(p,p')\}.$$
    Furthermore, $\ell' = \ell$ if $k>1$, otherwise
    $$\ell'((c_{1-i},p')) = \ell((p,p'))+\ell((c_{1-i},p)),$$ and $\ell'(e)=\ell(e)$ for all $e \neq (c_{1-i},p')$.
    \item Inverse edge contraction: inverse operation to edge contraction.
\end{itemize}
We also call edge contractions deletions and inverse edge contractions insertions, to have a more intuitive naming that also fits better to the classic edit operations on node-labeled trees. The three types of edit operations are illustrated on an example tree in Figure~\ref{fig:edit_operations_abstract}.
If a sequence of edit operations $S$ transforms an abstract merge tree $T_1$ into $T_2$, we denote this by $T_1 \xrightarrow{\scriptscriptstyle S} T_2$.

We define the costs of the edit operations between abstract merge trees as the change in persistence. As for classic edit distances, we denote the edit operations by pairs of two labels for relabel operations or pairs of a label and a blank symbol for deletions or insertions. We use $\mathbb{R}_{>0}$ as the label set for abstract merge trees and~$0$ as the blank symbol. Then, we define the cost function simply as the euclidean distance on $\mathbb{R}_{>0} \cup \{0\}$: $ \cost(l_1,l_2) = |l_1-l_2| $ for all $l_1,l_2 \in \mathbb{R}_{\geq 0}$.
This means, for a deletion or insertion we charge the persistence of the edge, whereas for a relabel we charge the persistence difference between the old and new edge.
We should note that abstract merge trees allow for other labels than persistence, but the cost function can be easily adapted to suit this usecase. For this paper, we restrict to persistence labels.

\subsubsection*{Model Relation}
The edit operations from the two models can be transformed into each other in an intuitive way. Since nodes in an abstract merge tree correspond to critical points of the continuous trees, shortening and extending operations are mapped to delete and insert operations if and only if at least one critical point of the continuous tree disappears and they are mapped to relabel operations otherwise. For the other direction, all edit operations on abstract merge trees are mapped to continuous operations in the obvious way.

Furthermore, the abstract operations can also be related to deformation retracts on continuous merge trees. The connection is the following:

\begin{theorem}
Let $\mathcal{T}_1,f_1$ and $\mathcal{T}_2,f_2$ be merge trees with abstractions $T_1=T(\mathcal{T}_1),\ell_{f_1}$ and $T_2=T(\mathcal{T}_2),\ell_{f_2}$. If $|\mathcal{T}_2|$ is homeomorphic to a deformation retract of $|\mathcal{T}_1|$, then there is a sequence of edit operations $S$ only containing deletions and relabels that decrease the edge labels such that $T_1,\ell_{f_1} \xrightarrow{\scriptscriptstyle S} T_2,\ell_{f_2}$.

Furthermore, given two abstract merge trees $T_1,\ell_1$ and $T_2,\ell_2$ with $T_1,\ell_1 \xrightarrow{\scriptscriptstyle S} T_2,\ell_2$ and $S$ only containing deletions and relabels that decrease the edge labels, then there are merge trees $\mathcal{T}_1,f_1$ and $\mathcal{T}_2,f_2$ with $T_1=T(\mathcal{T}_1),\ell_1=\ell_{f_1}$ and $T_2=T(\mathcal{T}_2),\ell_2=\ell_{f_2}$, such that $|\mathcal{T}_2|$ is homeomorphic to a deformation retract of $|\mathcal{T}_1|$.
\end{theorem}

\begin{proof}
See supplementary material, App.~A.
\end{proof}

The same also holds for insert operations and label-increasing relabels, with the adaption that $\mathcal{T}_1$ is a deformation retract of $\mathcal{T}_2$. This leads to another interesting property of the new edit distance: classic edit distances or edit mappings represent the largest common subtree, which means that our interpretation of the merge tree edit distances yields something like the largest common deformation retract, where the size of a merge tree is its total persistence.

\subsubsection*{Edit Distance}
The edit distance between two abstract merge trees $T_1,T_2$ is defined to be the minimal cost of an edit sequence transforming $T_1$ into $T_2$:
$$ \edist(T_1,T_2) = \min \{ \cost(S) \mid T_1 \xrightarrow{\scriptscriptstyle S} T_2 \}. $$
We call $S$ a one-degree edit sequence, if all insertions and deletions happen on edges connecting a leaf node. The edit distance based on these sequences is called one-degree edit distance:
$$ \odedist(T_1,T_2) = \min \{ \cost(S) \mid T_1 \xrightarrow{\scriptscriptstyle S} T_2,\ S \text{ is one-degree} \}. $$

Since edit sequences can be concatenated and the costs just add up, they are a metric for abstract merge trees. Since one-degree edit sequences do not restrict the possibility of concatenation, the one-degree edit distance is a metric, too.

\begin{theorem}

$\edist$ and $\odedist$ are metrics on the set of all abstract merge trees.

\end{theorem}

In~\cite{DBLP:journals/ipl/ZhangSS92} it has been shown that the problem of computing the general edit distance on unordered, node-labeled trees is NP-hard. Therefore, constrained versions like tree alignments~\cite{DBLP:conf/cpm/JiangWZ94} and the constrained edit distance~\cite{DBLP:journals/algorithmica/Zhang96} have been introduced. The one-degree edit distance~\cite{DBLP:journals/ipl/Selkow77} is a special case of the constrained edit distance and has tractable algorithms even for unordered trees. Due to the high complexity of unconstrained tree edit distances, merge tree edit distances are usually defined using one of the three kinds of constrained versions, see~\cite{DBLP:journals/cgf/SaikiaSW14,DBLP:journals/cgf/LohfinkWLWG20,DBLP:journals/tvcg/SridharamurthyM20,DBLP:journals/tvcg/PontVDT22,wetzels2022branch}. For the same reason, we will use $\odedist$ instead of $\edist$ throughout the rest of this paper. Intuitively, one-degree edit distances capture strongly connected subtrees instead of ancestor-preserving subtrees (which allow for gaps).

\begin{figure*}[!ht]
 
 \centering 
 \resizebox{\linewidth}{!}{
 \begin{tikzpicture}[yscale=0.75]
 \definecolor{lightred}{rgb}{1,0.6,0.6}
 
 
 \node[draw,circle,fill=gray!70] at (0, 1) (root) {\textbf{\Large A}};
 \node[draw,circle,fill=gray!70] at (0, 4) (s1) {\textbf{\Large B}};
 \node[draw,circle,fill=gray!70] at (-1, 6) (s2) {\textbf{\Large C}};
 
 \node[draw,circle,fill=red!70] at (-2, 10) (m1) {\textbf{\Large D}};
 \node[draw,circle,fill=red!70] at (0, 8) (m2) {\textbf{\Large E}};
 \node[draw,circle,fill=red!70] at (2, 9) (m3) {\textbf{\Large F}};
 
 \draw[blue,line width=3pt] (root) -- (s1);
 \draw[blue,line width=3pt] (s1) -- (s2);
 \draw[blue,line width=3pt] (s2) -- (m1);
 \draw[gray,line width=3pt] (s2) -- (m2);
 \draw[orange,line width=3pt] (s1) -- (m3);
 
 
 \node[draw,circle,fill=gray!70] at (6+0, 1) (root') {\textbf{\Large A}};
 \node[draw,circle,fill=gray!70] at (6+0, 4) (s1') {\textbf{\Large B}};
 
 \node[draw,circle,fill=red!70] at (6-2, 10) (m1') {\textbf{\Large D}};
 \node[draw,circle,fill=red!70] at (6+2, 9) (m3') {\textbf{\Large F}};
 
 \draw[blue,line width=3pt] (root') -- (s1');
 \draw[blue,line width=3pt] (s1') -- (m1');
 \draw[orange,line width=3pt] (s1') -- (m3');
 
 
 \node[draw,circle,fill=gray!70] at (13+0, 1) (root'') {\textbf{\Large A}};
 \node[draw,circle,fill=gray!70] at (13+0, 4) (s1'') {\textbf{\Large B}};
 \node[draw,circle,fill=gray!70] at (13-1, 6) (s2'') {\textbf{\Large C}};
 \node[draw,circle,fill=gray!70] at (13-1.5, 8) (s3'') {\textbf{\Large D}};
 
 \node[draw,circle,fill=red!70] at (13-2, 10) (m1'') {\textbf{\Large E}};
 \node[draw,circle,fill=red!70] at (13+0, 8) (m2'') {\textbf{\Large F}};
 \node[draw,circle,fill=red!70] at (13+2, 9) (m3'') {\textbf{\Large G}};
 \node[draw,circle,fill=red!70] at (13-1, 9.5) (m4'') {\textbf{\Large H}};
 
 \draw[green,line width=3pt] (root'') -- (s1'');
 \draw[blue,line width=3pt] (s1'') -- (s2'');
 \draw[blue,line width=3pt] (s2'') -- (s3'');
 \draw[blue,line width=3pt] (s3'') -- (m1'');
 \draw[gray,line width=3pt] (s3'') -- (m4'');
 \draw[gray,line width=3pt] (s2'') -- (m2'');
 \draw[orange,line width=3pt] (s1'') -- (m3'');
 
 
 \node[draw,circle,fill=gray!70] at (19+0, 1) (root'') {\textbf{\Large A}};
 \node[draw,circle,fill=gray!70] at (19+0, 4) (s1'') {\textbf{\Large B}};
 \node[draw,circle,fill=gray!70] at (19-1, 7) (s2'') {\textbf{\Large C}};
 
 \node[draw,circle,fill=red!70] at (19-2, 10) (m1'') {\textbf{\Large E}};
 \node[draw,circle,fill=red!70] at (19+0, 8.5) (m2'') {\textbf{\Large F}};
 \node[draw,circle,fill=red!70] at (19+2, 9) (m3'') {\textbf{\Large G}};
 
 \draw[green,line width=3pt] (root'') -- (s1'');
 \draw[blue,line width=3pt] (s1'') -- (s2'');
 \draw[blue,line width=3pt] (s2'') -- (m1'');
 \draw[gray,line width=3pt] (s2'') -- (m2'');
 \draw[orange,line width=3pt] (s1'') -- (m3'');
 
 
 \node[draw,circle,fill=gray!70] at (26+0, 1) (root'') {\textbf{\Large A}};
 \node[draw,circle,fill=gray!70] at (26+0, 4) (s1'') {\textbf{\Large B}};
 \node[draw,circle,fill=gray!70] at (26-1, 6) (s2'') {\textbf{\Large C}};
 \node[draw,circle,fill=gray!70] at (26-1.5, 8) (s3'') {\textbf{\Large D}};
 
 \node[draw,circle,fill=red!70] at (26-2, 10) (m1'') {\textbf{\Large E}};
 \node[draw,circle,fill=red!70] at (26+0, 8) (m2'') {\textbf{\Large F}};
 \node[draw,circle,fill=red!70] at (26+2, 9) (m3'') {\textbf{\Large G}};
 \node[draw,circle,fill=red!70] at (26-1, 9.5) (m4'') {\textbf{\Large H}};
 
 \draw[green,line width=3pt] (root'') -- (s1'');
 \draw[blue,line width=3pt] (s1'') -- (s2'');
 \draw[blue,line width=3pt] (s2'') -- (s3'');
 \draw[cyan,line width=3pt] (s3'') -- (m1'');
 \draw[gray,line width=3pt] (s3'') -- (m4'');
 \draw[gray,line width=3pt] (s2'') -- (m2'');
 \draw[orange,line width=3pt] (s1'') -- (m3'');
 
 
 \node[draw,circle,fill=gray!70] at (32+0, 1) (root'') {\textbf{\Large A}};
 \node[draw,circle,fill=gray!70] at (32+0, 4) (s1'') {\textbf{\Large B}};
 \node[draw,circle,fill=gray!70] at (32-1, 7) (s2'') {\textbf{\Large C}};
 
 \node[draw,circle,fill=red!70] at (32-2, 10) (m1'') {\textbf{\Large E}};
 \node[draw,circle,fill=red!70] at (32+0, 8.5) (m2'') {\textbf{\Large F}};
 \node[draw,circle,fill=red!70] at (32+2, 9) (m3'') {\textbf{\Large G}};
 
 \draw[green,line width=3pt] (root'') -- (s1'');
 \draw[blue,line width=3pt] (s1'') -- (s2'');
 \draw[cyan,line width=3pt] (s2'') -- (m1'');
 \draw[gray,line width=3pt] (s2'') -- (m2'');
 \draw[orange,line width=3pt] (s1'') -- (m3'');
 
 
 \node[] at (3, 3) (label1) {\Huge\textbf{(a)}};
 \node[] at (16, 3) (label1) {\Huge\textbf{(b)}};
 \node[] at (29, 3) (label1) {\Huge\textbf{(c)}};

 \end{tikzpicture}
 }
 
 \caption{Three examples for valid and invalid path mappings. Mapped paths are indicated through the coloring. In (a), an invalid mapping is shown, which contradicts Lemma~\ref{lemma:simplePathProperty}, since condition~3 is not fulfilled for the path $(B,F)$. A valid mapping can be seen in (b). The mapping in~(c) fulfills the property in Lemma~\ref{lemma:simplePathProperty}, but not Lemma~\ref{lemma:optPathMappingsAreAbstractMergeTrees}, since the paths $(B,D)$ and $(D,E)$ can be merged to $(B,E)$ without increasing the cost.}
\label{fig:pathLemma}
 
\end{figure*}
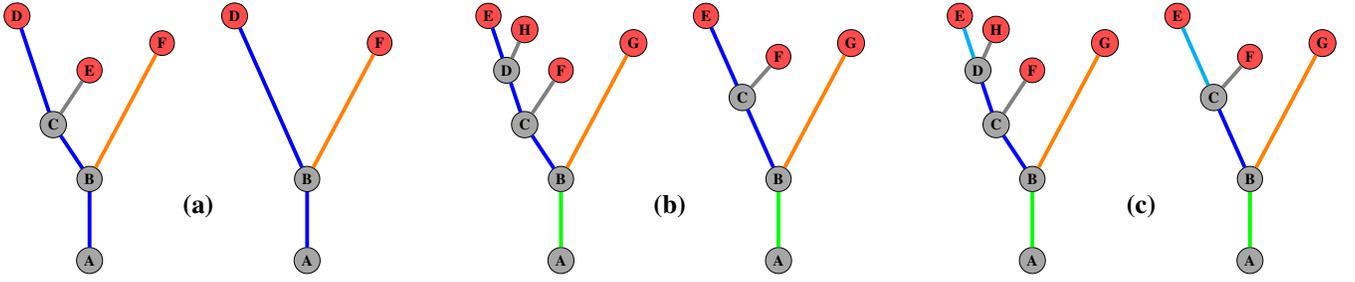

\subsection{Path Mappings}
\label{section:pathmappings}

Edit sequences using the edit operations defined in the last section induce mappings between abstract merge trees in a similar way as edit sequences for the classic tree edit distance do. However, in contrast to mappings between the edges or nodes of the trees, we get mappings between the paths of two abstract merge trees. To see why, consider the definition of an edge contraction. If the remaining node has only one child, we remove it from the tree and connect its only child to its parent. The new edge in the resulting tree is created by merging two edges in the original tree. This relation is represented by the mapping. Hence, the new edge in the resulting tree is mapped to a path of length~$2$ in the original tree, consisting of the two edges that were merged. This correspondence also works transitively for multiple operations and also for inverse edge contractions the other way around. Figure~\ref{fig:pathMappingIllustration} illustrates this correspondence and how to derive a path mapping from a sequence of edit operations. Next, we will study path mappings in a more formal manner, starting with a definition.

\begin{definition}
\label{definiton:path_mappings}

Given two abstract merge trees $T_1,T_2$, a path mapping between $T_1$ and $T_2$ is a mapping $M \subseteq \mathcal{P}(T_1) \times \mathcal{P}(T_2)$ such that
\begin{enumerate}
    \item $p_1=p_2$ if and only if $q_1=q_2$ for all $(p_1,q_1),(p_2,q_2) \in M$,
    \item $|p_1 \cap p_2| \leq 1$ and $|q_1 \cap q_2| \leq 1$ for all $(p_1,q_1),(p_2,q_2) \in M$,
    \item for all $(p,q) \in M$,
    \begin{itemize}
        \item either there are paths $p' \in \mathcal{P}(T_1)$ and $q' \in \mathcal{P}(T_2)$ such that $(p',q') \in M$ and $\pathstart(p) = \pathend(p')$ and $\pathstart(q) = \pathend(q')$, 
        \item or $\pathstart(p) = \troot(T_1)$ and $\pathstart(q) = \troot(T_2)$.
    \end{itemize}
\end{enumerate}
For a path mapping $M$ between two abstract merge trees $T_1,\ell_1$ and $T_2,\ell_2$, we also define its corresponding edit operations $\mappingedit(M)$. They consist of the corresponding relabel, insert and delete operations:
\[
\begin{aligned}
\mappingrelabel(M) &= \{ (\ell_1(p_1),\ell_2(p_2)) \mid (p_1,p_2) \in M \}, \\
\mappinginsert(M) &= \{ (0,\ell_2(e_2)) \mid e_2 \in E(T_2),\ \nexists p_1 \in \mathcal{P}(T_1): (p_1,e_2) \in M \}, \\
\mappingdelete(M) &= \{ (\ell_1(e_1),0) \mid e_1 \in E(T_1),\ \nexists p_2 \in \mathcal{P}(T_2): (e_1,p_2) \in M \}.
\end{aligned}
\]
Then we have $\mappingedit(M) = \mappingrelabel(M) \cup \mappinginsert(M) \cup \mappingdelete(M)$.
Furthermore, we define the costs of a mapping through the corresponding edit operations:
$$ \cost(M) = \sum_{(l_1,l_2) \in \mappingedit(M)} |\cost(l_1,l_2)| $$

\end{definition}

We say a path $p_1 \in \mathcal{P}(T_1)$ is contained in a path mapping $M \subseteq \mathcal{P}(T_1) \times \mathcal{P}(T_2)$, if there is a path $p_2 \in \mathcal{P}(T_2)$ with \mbox{$(p_1,p_2) \in M$.} Also, $p_2$ is contained in $M$ if the symmetrical condition holds.
We say the subtree $T_1[(u,v)]$ is not contained in $M$ if for any $p' \in \mathcal{P}(T_1[(u,v)])$, $p'$ is not contained in $M$.

A core property of \emph{optimal} path mappings is that the contained paths are not unnecessarily partitioned, i.e.\ if we interpret them as edges, these edges form an abstract merge tree without degree one nodes (except the root). Intuitively, the reason for this can be seen from the way paths are mapped in Figure~\ref{fig:pathMappingIllustration}: all subtrees branching from mapped paths ($ABCD,ABC,ACD$) are either deleted or inserted, i.e.\ they are not present in the mapping. In other words, mapped branches do not start within other mapped paths. This property is given in Lemma~\ref{lemma:simplePathProperty}. Furthermore, all mapped paths either end in a leaf node (e.g.\ $ABCD$) or end in the starting vertex of at least two other mapped paths (e.g.\ $ABC$ splits into $CD$ amd $ED$, both also present in the mapping). This is formalized in Lemma~\ref{lemma:optPathMappingsAreAbstractMergeTrees}. The properties are also illustrated in Figure~\ref{fig:pathLemma}. We now discuss this more formally.

First, we note that the following property directly follows from the definition of path mappings (conditions 2 and 3).

\begin{lemma}
\label{lemma:simplePathProperty}

Let $M \subseteq \mathcal{P}(T_1) \times \mathcal{P}(T_2)$ be an optimal path mapping between two abstract merge trees $T_1,T_2$. For any path \mbox{$p = v_1...v_k \in \mathcal{P}(T_1)$} that is contained in $M$ and any child $c_i \neq v_{i+1}$ of $v_i$ ($1 < i < k$), the subtree $T_1[(c_i,v_i)]$ is not contained in $M$.

Symmetrically, the same holds for any path $p \in \mathcal{P}(T_2)$.

\end{lemma}

Furthermore, any path contained in an optimal path mapping branches into at least two other contained branches. A proof for this claim is provided in the supplementary material.

\begin{lemma}
\label{lemma:optPathMappingsAreAbstractMergeTrees}

Let $M \subseteq \mathcal{P}(T_1) \times \mathcal{P}(T_2)$ be an optimal path mapping between two abstract merge trees $T_1,T_2$. For any path $p \in \mathcal{P}(T_1)$ that is contained in $M$, there are at least two paths $p',p''$ that are contained in $M$ with $\pathstart(p') = \pathstart(p'') = \pathend(p)$.

Symmetrically, the same holds for any path $p \in \mathcal{P}(T_2)$.

\end{lemma}

\begin{proof}
See supplementary material, App.~B.
\end{proof}




The properties from the last two lemmas are illustrated Figure~\ref{fig:pathLemma}. As a next step, we now focus on the equivalence of path mappings and the here defined one-degree edit distance for abstract merge trees. We show the equivalence by first proving that an optimal edit sequence has a corresponding mapping of lower or equal cost and secondly that each mapping has a corresponding edit sequence. This then allows us to compute optimal path mappings instead of optimal edit sequence.

\begin{lemma}
\label{lemma:editToMapping}

Let $S$ be a cost-optimal sequence of edit operations that transforms an abstract merge tree $T_1$ into another one $T_2$. Then there exists a path mapping $M \subseteq \mathcal{P}(T_1) \times \mathcal{P}(T_2)$ such that $\cost(M) \leq \cost(S)$.

\end{lemma}

\begin{proof}
See supplementary material, App.~C.
\end{proof}

\begin{figure*}[]
\begin{subfigure}[t]{0.24\linewidth}
    \centering
    \resizebox{0.99\linewidth}{!}{
    \begin{tikzpicture}[xscale=0.6,yscale=0.8]
    
    \node[] at (0, 0) (p1) {\tiny$p_1$};
    \node[] at (0, -1) (n1) {\tiny$n_1$};
    \node[] at (-1,-2) (c11) {\tiny$c_{1,1}$};
    \node[] at (1, -2) (c12) {\tiny$c_{1,2}$};
    \draw [ultra thick,ForestGreen] (p1) -- (n1);
    \draw [ultra thick,ProcessBlue] (n1) -- (c11);
    \draw [ultra thick,RedOrange] (n1) -- (c12);
    \draw [very thick,black] (c11) -- (-1.5,-2.4) -- (-1.75,-2.75) -- (-1.25,-2.75) -- (-1.5,-2.4);
    \draw [very thick,black] (c11) -- (-0.5,-2.4) -- (-0.75,-2.75) -- (-0.25,-2.75) -- (-0.5,-2.4);
    \draw [very thick,black] (c12) -- (1.5,-2.4) -- (1.75,-2.75) -- (1.25,-2.75) -- (1.5,-2.4);
    \draw [very thick,black] (c12) -- (0.5,-2.4) -- (0.75,-2.75) -- (0.25,-2.75) -- (0.5,-2.4);
    
    \node[] at (0+4, 0) (p2) {\tiny$p_2$};
    \node[] at (0+4, -1) (n2) {\tiny$n_2$};
    \node[] at (-1+4, -2) (c21) {\tiny$c_{2,1}$};
    \node[] at (1+4, -2) (c22) {\tiny$c_{2,2}$};
    \draw [ultra thick,ForestGreen] (p2) -- (n2);
    \draw [ultra thick,ProcessBlue] (n2) -- (c21);
    \draw [ultra thick,RedOrange] (n2) -- (c22);
    \draw [very thick,black] (c21) -- (-1.5+4,-2.4) -- (-1.75+4,-2.75) -- (-1.25+4,-2.75) -- (-1.5+4,-2.4);
    \draw [very thick,black] (c21) -- (-0.5+4,-2.4) -- (-0.75+4,-2.75) -- (-0.25+4,-2.75) -- (-0.5+4,-2.4);
    \draw [very thick,black] (c22) -- (1.5+4,-2.4) -- (1.75+4,-2.75) -- (1.25+4,-2.75) -- (1.5+4,-2.4);
    \draw [very thick,black] (c22) -- (0.5+4,-2.4) -- (0.75+4,-2.75) -- (0.25+4,-2.75) -- (0.5+4,-2.4);
    
    \node[] at (2,-0.3) (label) {\textbf{(a)}};
    
    \end{tikzpicture}
    }
    \caption*{\centering\mbox{$d_P(c_{1,1},p_1,c_{2,1},p_2)+d_P(c_{1,2},n_1,c_{2,2},n_2)$} $+\cost(p_1...n_1,p_2...n_2)$}
\end{subfigure}
\begin{subfigure}[t]{0.24\linewidth}
    \centering
    \resizebox{0.99\linewidth}{!}{
    \begin{tikzpicture}[xscale=0.6,yscale=0.8]
    
    \node[] at (0, 0) (p1) {\tiny$p_1$};
    \node[] at (0, -1) (n1) {\tiny$n_1$};
    \node[] at (-1,-2) (c11) {\tiny$c_{1,1}$};
    \node[] at (1, -2) (c12) {\tiny$c_{1,2}$};
    \draw [ultra thick,ForestGreen] (p1) -- (n1);
    \draw [ultra thick,ProcessBlue] (n1) -- (c11);
    \draw [ultra thick,RedOrange] (n1) -- (c12);
    \draw [very thick,black] (c11) -- (-1.5,-2.4) -- (-1.75,-2.75) -- (-1.25,-2.75) -- (-1.5,-2.4);
    \draw [very thick,black] (c11) -- (-0.5,-2.4) -- (-0.75,-2.75) -- (-0.25,-2.75) -- (-0.5,-2.4);
    \draw [very thick,black] (c12) -- (1.5,-2.4) -- (1.75,-2.75) -- (1.25,-2.75) -- (1.5,-2.4);
    \draw [very thick,black] (c12) -- (0.5,-2.4) -- (0.75,-2.75) -- (0.25,-2.75) -- (0.5,-2.4);
    
    \node[] at (0+4, 0) (p2) {\tiny$p_2$};
    \node[] at (0+4, -1) (n2) {\tiny$n_2$};
    \node[] at (-1+4, -2) (c21) {\tiny$c_{2,1}$};
    \node[] at (1+4, -2) (c22) {\tiny$c_{2,2}$};
    \draw [ultra thick,ForestGreen] (p2) -- (n2);
    \draw [ultra thick,RedOrange] (n2) -- (c21);
    \draw [ultra thick,ProcessBlue] (n2) -- (c22);
    \draw [very thick,black] (c21) -- (-1.5+4,-2.4) -- (-1.75+4,-2.75) -- (-1.25+4,-2.75) -- (-1.5+4,-2.4);
    \draw [very thick,black] (c21) -- (-0.5+4,-2.4) -- (-0.75+4,-2.75) -- (-0.25+4,-2.75) -- (-0.5+4,-2.4);
    \draw [very thick,black] (c22) -- (1.5+4,-2.4) -- (1.75+4,-2.75) -- (1.25+4,-2.75) -- (1.5+4,-2.4);
    \draw [very thick,black] (c22) -- (0.5+4,-2.4) -- (0.75+4,-2.75) -- (0.25+4,-2.75) -- (0.5+4,-2.4);
    
    \node[] at (2,-0.3) (label) {\textbf{(b)}};
    
    \end{tikzpicture}
    }
    \caption*{\centering\mbox{$d_P(c_{1,1},p_1,c_{2,2},p_2)+d_P(c_{1,2},n_1,c_{2,1},n_2)$} $+\cost(p_1...n_1,p_2...n_2)$}
\end{subfigure}
\begin{subfigure}[t]{0.24\linewidth}
    \centering
    \resizebox{0.99\linewidth}{!}{
    \begin{tikzpicture}[xscale=0.6,yscale=0.8]
    
    \node[] at (0, 0) (p1) {\tiny$p_1$};
    \node[] at (0, -1) (n1) {\tiny$n_1$};
    \node[] at (-1,-2) (c11) {\tiny$c_{1,1}$};
    \node[] at (1, -2) (c12) {\tiny$c_{1,2}$};
    \draw [ultra thick,ForestGreen] (p1) -- (n1);
    \draw [ultra thick,ForestGreen] (n1) -- (c11);
    \draw [ultra thick,gray!40] (n1) -- (c12);
    \draw [very thick,black] (c11) -- (-1.5,-2.4) -- (-1.75,-2.75) -- (-1.25,-2.75) -- (-1.5,-2.4);
    \draw [very thick,black] (c11) -- (-0.5,-2.4) -- (-0.75,-2.75) -- (-0.25,-2.75) -- (-0.5,-2.4);
    \draw [very thick,gray!40] (c12) -- (1.5,-2.4) -- (1.75,-2.75) -- (1.25,-2.75) -- (1.5,-2.4);
    \draw [very thick,gray!40] (c12) -- (0.5,-2.4) -- (0.75,-2.75) -- (0.25,-2.75) -- (0.5,-2.4);
    
    \node[] at (0+4, 0) (p2) {\tiny$p_2$};
    \node[] at (0+4, -1) (n2) {\tiny$n_2$};
    \node[] at (-1+4, -2) (c21) {\tiny$c_{2,1}$};
    \node[] at (1+4, -2) (c22) {\tiny$c_{2,2}$};
    \draw [ultra thick,ForestGreen] (p2) -- (n2);
    \draw [ultra thick,black] (n2) -- (c21);
    \draw [ultra thick,black] (n2) -- (c22);
    \draw [very thick,black] (c21) -- (-1.5+4,-2.4) -- (-1.75+4,-2.75) -- (-1.25+4,-2.75) -- (-1.5+4,-2.4);
    \draw [very thick,black] (c21) -- (-0.5+4,-2.4) -- (-0.75+4,-2.75) -- (-0.25+4,-2.75) -- (-0.5+4,-2.4);
    \draw [very thick,black] (c22) -- (1.5+4,-2.4) -- (1.75+4,-2.75) -- (1.25+4,-2.75) -- (1.5+4,-2.4);
    \draw [very thick,black] (c22) -- (0.5+4,-2.4) -- (0.75+4,-2.75) -- (0.25+4,-2.75) -- (0.5+4,-2.4);
    
    \node[] at (2,-0.3) (label) {\textbf{(c)}};
    
    \end{tikzpicture}
    }
    \caption*{\mbox{$d_P(c_{1,1},p_1,n_1,p_2)+d_P(c_{1,2},n_1,\bot,\bot)$}}
\end{subfigure}
\begin{subfigure}[t]{0.24\linewidth}
    \centering
    \resizebox{0.99\linewidth}{!}{
    \begin{tikzpicture}[xscale=0.6,yscale=0.8]
    
    \node[] at (0, 0) (p1) {\tiny$p_1$};
    \node[] at (0, -1) (n1) {\tiny$n_1$};
    \node[] at (-1,-2) (c11) {\tiny$c_{1,1}$};
    \node[] at (1, -2) (c12) {\tiny$c_{1,2}$};
    \draw [ultra thick,ForestGreen] (p1) -- (n1);
    \draw [ultra thick,gray!40] (n1) -- (c11);
    \draw [ultra thick,ForestGreen] (n1) -- (c12);
    \draw [very thick,gray!40] (c11) -- (-1.5,-2.4) -- (-1.75,-2.75) -- (-1.25,-2.75) -- (-1.5,-2.4);
    \draw [very thick,gray!40] (c11) -- (-0.5,-2.4) -- (-0.75,-2.75) -- (-0.25,-2.75) -- (-0.5,-2.4);
    \draw [very thick,black] (c12) -- (1.5,-2.4) -- (1.75,-2.75) -- (1.25,-2.75) -- (1.5,-2.4);
    \draw [very thick,black] (c12) -- (0.5,-2.4) -- (0.75,-2.75) -- (0.25,-2.75) -- (0.5,-2.4);
    
    \node[] at (0+4, 0) (p2) {\tiny$p_2$};
    \node[] at (0+4, -1) (n2) {\tiny$n_2$};
    \node[] at (-1+4, -2) (c21) {\tiny$c_{2,1}$};
    \node[] at (1+4, -2) (c22) {\tiny$c_{2,2}$};
    \draw [ultra thick,ForestGreen] (p2) -- (n2);
    \draw [ultra thick,black] (n2) -- (c21);
    \draw [ultra thick,black] (n2) -- (c22);
    \draw [very thick,black] (c21) -- (-1.5+4,-2.4) -- (-1.75+4,-2.75) -- (-1.25+4,-2.75) -- (-1.5+4,-2.4);
    \draw [very thick,black] (c21) -- (-0.5+4,-2.4) -- (-0.75+4,-2.75) -- (-0.25+4,-2.75) -- (-0.5+4,-2.4);
    \draw [very thick,black] (c22) -- (1.5+4,-2.4) -- (1.75+4,-2.75) -- (1.25+4,-2.75) -- (1.5+4,-2.4);
    \draw [very thick,black] (c22) -- (0.5+4,-2.4) -- (0.75+4,-2.75) -- (0.25+4,-2.75) -- (0.5+4,-2.4);
    
    \node[] at (2,-0.3) (label) {\textbf{(d)}};
    
    \end{tikzpicture}
    }
    \caption*{\mbox{$d_P(c_{1,2},p_1,n_1,p_2)+d_P(c_{1,1},n_1,\bot,\bot)$}}
\end{subfigure}
\caption{Exemplary illustration of recursive structure of the optimal path mapping of the trees rooted in $(n_1,p_1)$ and $(n_2,p_2)$, denoted by $d_P(n_1,p_1,n_2,p_2)$. In (a) and (b) we stop the currently tracked path (green) and start to track two new paths (blue and red) in both trees. Which of these are mapped is shown by the colors. In (c) we continue the current path in $T_1$ with the left child and delete the right subtree, while staying at the same node in $T_2$. We do the same recursion with left and right child swapped in (d).}
\label{fig:recursion_dp}
\end{figure*}
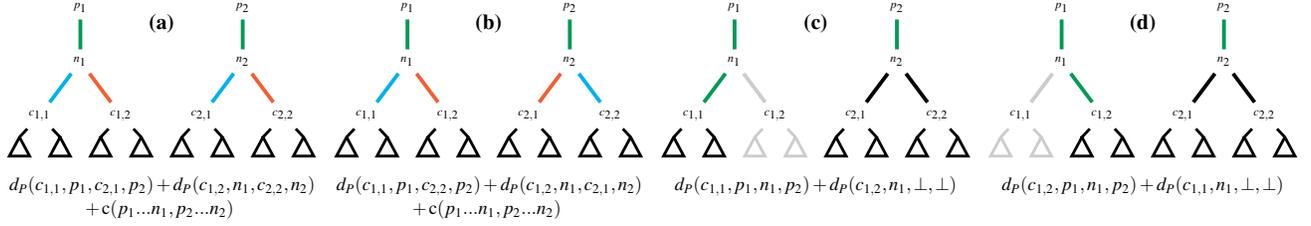

\begin{lemma}
\label{lemma:mappingToEdit}

For two abstract merge trees $T_1,T_2$, let $M \subseteq \mathcal{P}(T_1) \times \mathcal{P}(T_2)$ be a path mapping. Then there exists a sequence $S$ of edit operations that transforms $T_1$ into $T_2$, with $\cost(S) = \cost(M)$.

\end{lemma}

\begin{proof}
See supplementary material, App.~D.
\end{proof}

Now, we can conclude this part of the paper with the core theorem concerning path mappings. From here on, we will use the term \emph{path mapping distance} equivalently to $\odedist$.

\begin{theorem}
\label{theorem:editToMapping}

For two abstract merge trees $T_1,T_2$, the one-degree edit distance is exactly the cost of an optimal path mapping between $T_1$ and $T_2$:
$$ \odedist(T_1,T_2) = \min \{ \cost(M) \mid M \text{ is a path mapping between } T_1 \text{ and } T_2 \}. $$

\end{theorem}

\begin{proof}
Follows directly from Lemmas~\ref{lemma:editToMapping} and~\ref{lemma:mappingToEdit}.
\end{proof}

\subsection{Recursive Structure}

We now know that the one-degree edit distance is equivalent to optimal path mappings. Next, we will investigate the recursive structure of path mappings, which can then be exploited to obtain efficient polynomial time algorithms.
We omit a formal discussion of the base cases since the optimal path mapping between a non-empty abstract merge tree and an empty one is of course the empty mapping, and for two trees with just one edge, we always map the two unique edges.
Furthermore, we only consider binary merge trees for simplicity. The recursion can of course be easily adapted for trees of arbitrary degree.

\begin{lemma} Given two abstract merge trees $T_1,T_2$ with $\troot(T_1)=v_1$ and $\troot(T_2) = u_1$, let $v_2,u_2$ be the unique children of the two roots and let those have children $v_3,v_4$ and $u_3,u_4$. Let $T_1' = T_1[(v_2,v_3)]$, $T_1'' = T_1[(v_2,v_4)]$, $T_2' = T_2[(u_2,u_3)]$ and $T_2'' = T_2[(u_2,u_4)]$. Then, for the one-degree edit distance between $T_1$ and $T_2$, it holds that $\odedist(T_1,T_2)$ is either
\begin{itemize}
    \item $\odedist(T_1',\bot) + \odedist(T_1-T_1',T_2)$ or
    \item $\odedist(\bot,T_2') + \odedist(T_1,T_2-T_2')$ or
    \item $\odedist(T_1'',\bot) + \odedist(T_1-T_1'',T_2)$ or
    \item $\odedist(\bot,T_2'') + \odedist(T_1,T_2-T_2'')$ or
    \item $\odedist(T_1',T_2') + \odedist(T_1'',T_2'') + \cost(\ell_1((v_2,v_1)),\ell_2((u_2,u_1)))$ or
    \item $\odedist(T_1',T_2'') + \odedist(T_1'',T_2') + \cost(\ell_1((v_2,v_1)),\ell_2((u_2,u_1)))$.
\end{itemize}
\label{lemma:rec_inner}
\end{lemma}

\begin{proof}
To show this recursion, we consider the optimal path mapping $M$ between $T_1$ and $T_2$. We know that $M$ is not empty, since both $T_1$ and $T_2$ are non-empty trees. Furthermore, we know that there is a pair of paths $p_1 \in \mathcal{P}(T_1),p_2 \in \mathcal{P}(T_2)$ in the mapping, $(p_1,p_2) \in M$, such that both begin at the roots of the two trees, i.e.\ $\pathstart(p_1) = v_1$ and $\pathstart(p_2) = u_1$. Then, we can make a distinction between two cases: either $(v_1v_2,u_1u_2) \in M$ \textbf{(a)} or $(v_1v_2,u_1u_2) \notin M$~\textbf{(b)}. Or in other words, either $p_1=v_1v_2,p_2=u_1u_2$ holds or not.

\textbf{(a)} First, we consider the case that $(v_1v_2,u_1u_2) \in M$. By Lemma~\ref{lemma:optPathMappingsAreAbstractMergeTrees}, we know that there are at least two paths $p'_1,p''_1 \in \mathcal{P}(T_1)$ contained in $M$ with $\pathstart(p'_1) = \pathstart(p''_1) = v_2$.
Since $p'_1$ and $p''_1$ go through $v_3$ and $v_4$ and start in $v_2$, we can restrict $M$ to the vertices and edges of $T_1' = T_1[(v_2,v_3)]$ and $T_1'' = T_1[(v_2,v_4)]$ an obtain again two path mappings $M',M''$. To see that $M',M''$ are indeed path mappings, we only have to check condition~3 for $p'_1$ and $p''_1$, since only for those two paths the parent paths are removed. However, as they both start in the roots of the corresponding trees, condition~3 is still fulfilled.
Hence, we get that
$$\odedist(T_1,T_2) = \odedist(T_1',T_2') + \odedist(T_1'',T_2'') + \cost(\ell_1((v_2,v_1)),\ell_2((u_2,u_1)))$$
if $p'_1$ goes through $v_3$ and $M(p'_1)$ through $u_3$, or
$$\odedist(T_1,T_2) = \odedist(T_1',T_2'') + \odedist(T_1'',T_2') + \cost(\ell_1((v_2,v_1)),\ell_2((u_2,u_1)))$$
otherwise, i.e.\ if $p'_1$ goes through $v_3$ and $M(p'_1)$ through $u_4$.

\textbf{(b)} In the second case that $(v_1v_2,u_1u_2) \notin M$, we know that $v_1v_2$ is not contained in $M$ or $u_1u_2$ is not contained in $M$, which equivalently means that there is a path $v_1v_2v'_1...v'_k \in \mathcal{P}(T_1)$ contained in $M$ with $k \geq 1$ or there is a path $u_1u_2u'_1...u'_k \in \mathcal{P}(T_2)$ contained in $M$ with $k \geq 1$.
If $v_1v_2v'_1...v'_k$ is contained in $M$ and $v'_1 = v_3$, then $T_1'' = T_1[(v_2,v_4)]$ is not contained in $M$ according to Lemma~\ref{lemma:simplePathProperty}. Let $M''$ be an optimal path mapping between $T_1''$ and the empty tree and $M'$ be an optimal path mapping between $T_1 - T_1''$ and $T_2$. We have $\mappingdelete(M'') \subseteq \mappingdelete(M)$ since $T_1''$ is not contained in $M$, and $\mappingedit(M') = \mappingedit(M) \setminus \mappingedit(M'')$ since $\ell_1(v_1v_2v'_1...v'_k) = \ell_1(v_1v'_1...v'_k)$. Both together give us
$$\odedist(T_1,T_2) = \odedist(T_1'',\bot) + \odedist(T_1-T_1'',T_2).$$
The other three cases, $v'_1 = v_4$, $u'_1 = u_3$ and $u'_1 = u_4$ all work symmetrically and in total, we get the the six cases from the lemma to show.
\end{proof}

This recursion gives rise to a dynamic programming algorithm which we discuss in the next section.

\SetKwComment{Comment}{/* }{ */}
\setlength{\algomargin}{5pt}
\SetAlgoVlined
\begin{algorithm}[!ht]
\caption{Computing the path mapping distance}
\label{alg:odedist}
\SetKwFunction{pathDist}{$\odedist$}
\SetKwProg{Fn}{Function}{:}{}
\DontPrintSemicolon
\Fn{\pathDist{$n_1,p_1,n_2,p_2$}}{
\If{$n_1=\bot$ and $n_2$ is a leaf}{
  \Return $\cost(\bot,p_2...n_2)$\;
}
\If{$n_2=\bot$ and $n_1$ is a leaf}{
  \Return $\cost(p_1...n_1,\bot)$\;
}
\If{$n_1$ is a leaf and $n_2$ is a leaf}{
  \Return $\cost(p_1...n_1,p_2...n_2)$\;
}
\If{$n_1=\bot$ and $n_2$ is an inner node}{
  Let $c_{2,1},c_{2,2}$ be the children of $n_2$\;
  \Return~$\pathDist(\bot,\bot,c_{2,1},n_2) + \pathDist(\bot,\bot,c_{2,2},n_2) + \cost(\bot,p_2...n_2)$\;
}
\If{$n_2=\bot$ and $n_1$ is an inner node}{
  Let $c_{1,1},c_{1,2}$ be the children of $n_1$\;
  \Return~$\pathDist(c_{1,1},n_1,\bot,\bot) + \pathDist(c_{1,2},n_1,\bot,\bot) + \cost(p_1...n_1,\bot)$\;
}
\If{$n_1$ is a leaf and $n_2$ is an inner node}{
  Let $c_{2,1},c_{2,2}$ be the children of $n_2$\;
  \Return~$\min \begin{cases}
           \pathDist(n_1,p_1,c_{2,1},p_2) + \pathDist(\bot,\bot,c_{2,2},n_2)\\
           \pathDist(n_1,p_1,c_{2,2},p_2) + \pathDist(\bot,\bot,c_{2,1},n_2)
           \end{cases}$\;
  \vspace{-10pt}
}
\If{$n_2$ is a leaf and $n_1$ is an inner node}{
  Let $c_{1,1},c_{1,2}$ be the children of $n_1$\;
  \Return~$\min \begin{cases}
           \pathDist(c_{1,1},p_1,n_2,p_2) + \pathDist(c_{1,2},n_1,\bot,\bot)\\
           \pathDist(c_{1,2},p_1,n_2,p_2) + \pathDist(c_{1,1},n_1,\bot,\bot)
           \end{cases}$\;
  \vspace{-10pt}
}
\If{$n_1$ is an inner node and $n_2$ is an inner node}{
  Let $c_{1,1},c_{1,2}$ be the children of $n_1$\;
  Let $c_{2,1},c_{2,2}$ be the children of $n_2$\;
  $d_1 \coloneqq \min \begin{cases}
           \pathDist(c_{1,1},p_1,n_2,p_2) + \pathDist(c_{1,2},n_1,\bot,\bot)\\
           \pathDist(c_{1,2},p_1,n_2,p_2) + \pathDist(c_{1,1},n_1,\bot,\bot)
           \end{cases}$\;
  $d_2 \coloneqq \min \begin{cases}
           \pathDist(n_1,p_1,c_{2,1},p_2) + \pathDist(\bot,\bot,c_{2,2},n_2)\\
           \pathDist(n_1,p_1,c_{2,2},p_2) + \pathDist(\bot,\bot,c_{2,1},n_2)
           \end{cases}$\;
  \mbox{$d_3 \coloneqq \min \begin{cases}
           \pathDist(c_{1,1},n_1,c_{2,1},n_2) + \pathDist(c_{1,2},n_1,c_{2,2},n_2)\\
           \pathDist(c_{1,1},n_1,c_{2,1},n_2) + \pathDist(c_{1,2},n_1,c_{2,2},n_2)
           \end{cases}$}\;
  \vspace{-10pt}
  \Return~$\min(d_1,d_2,d_3+\cost(p_1...n_1,p_2...n_2))$\;
}
}
\end{algorithm}

\section{Implementation and Experiments}
\label{section:experiments}

\begin{figure*}[t]
\centering
\includegraphics[width=0.25\linewidth]{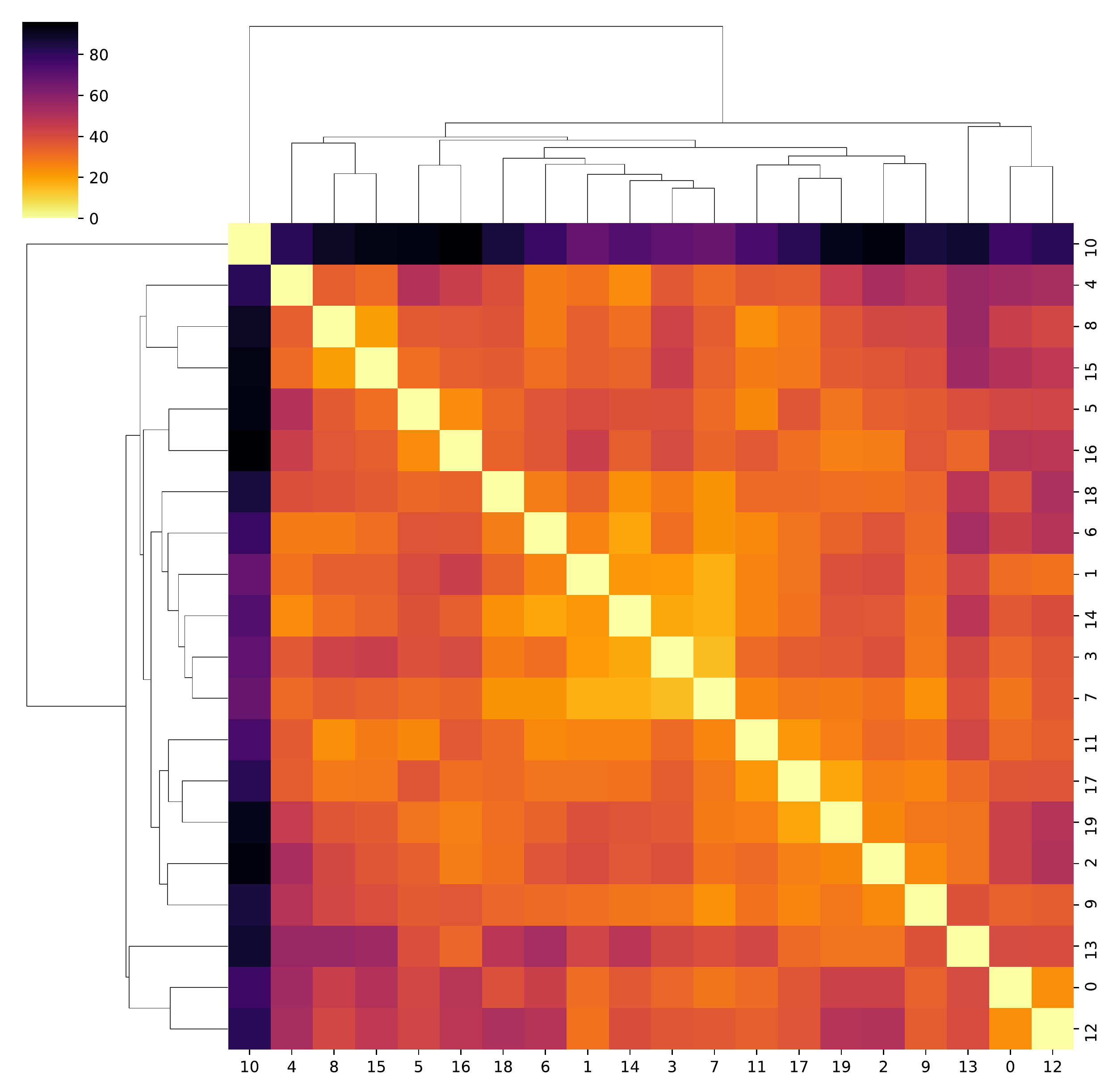}
\hspace{0.05\linewidth}
\includegraphics[width=0.25\linewidth]{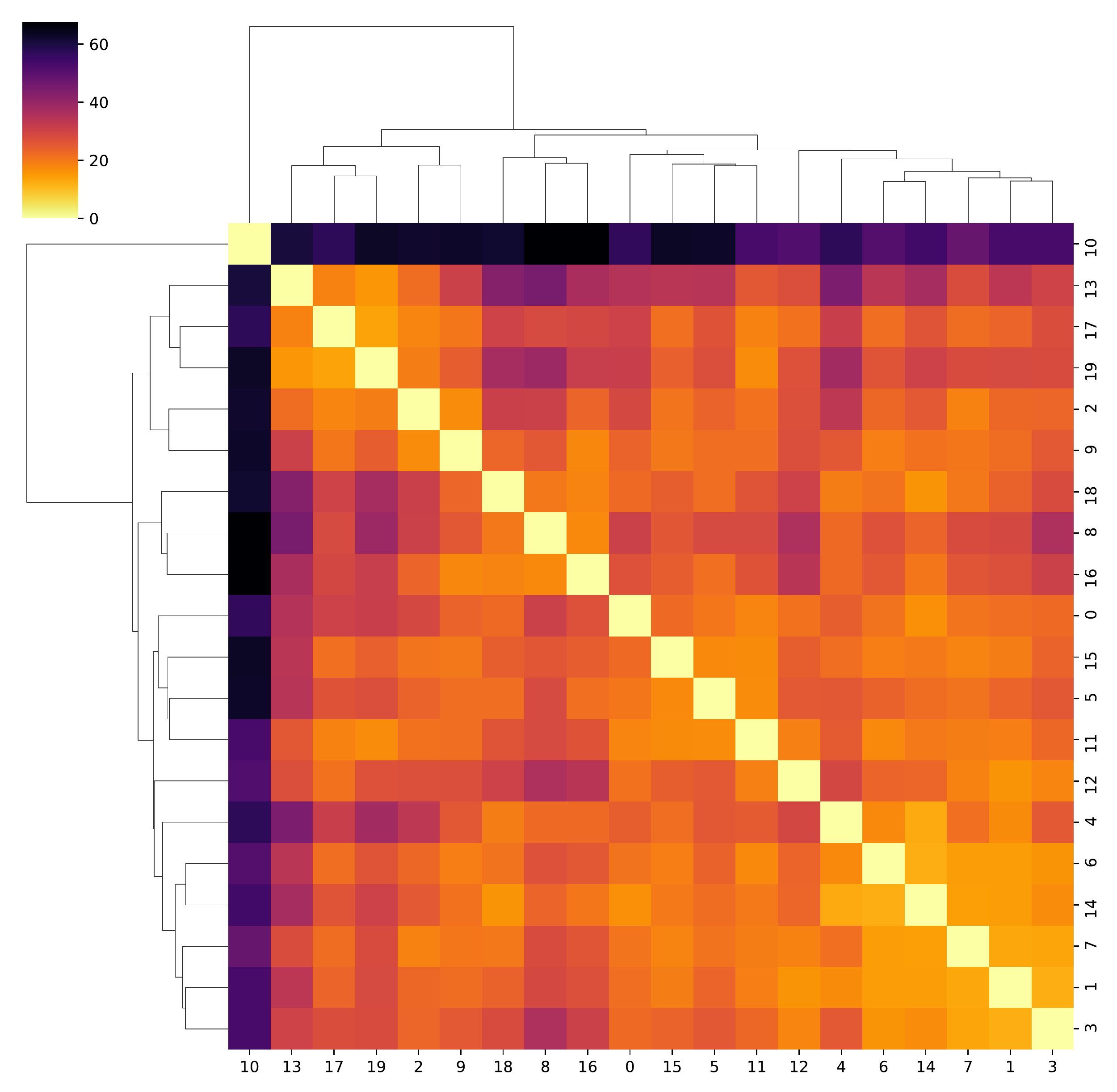}
\hspace{0.05\linewidth}
\includegraphics[width=0.25\linewidth]{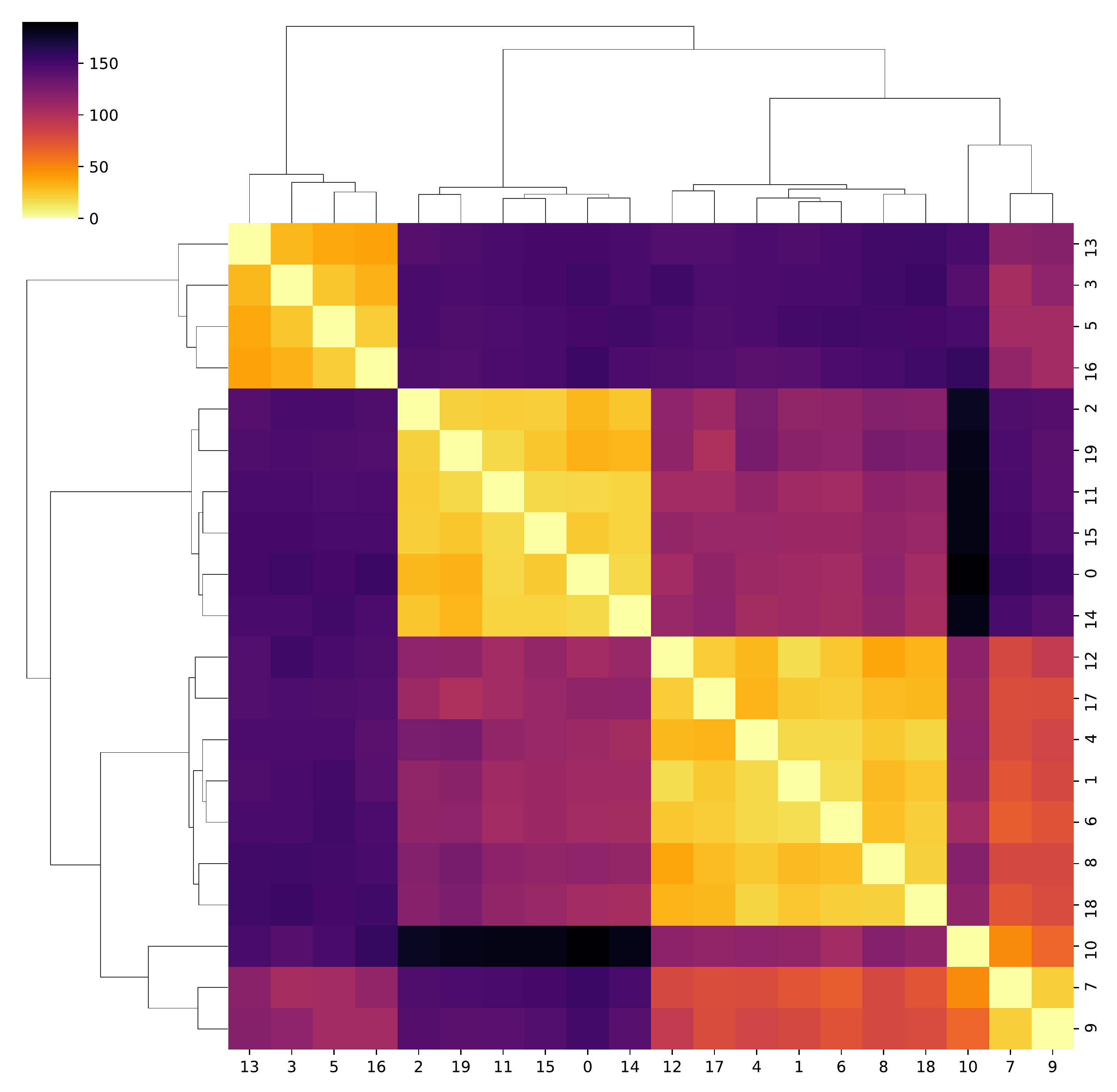}
\caption{The clustermaps with dendograms for the outlier ensemble using the path mapping distance (left), the branch mapping distance (middle) and the classic constrained edit distance (right).}
\label{fig:dm_outlier2_path}
\end{figure*}

We now present an $\mathcal{O}(n^4)$ algorithm for computing $\odedist$ on binary abstract merge trees. It strongly resembles the dynamic programming for branch mappings in~\cite{wetzels2022branch}. Although branch mappings and path mappings differ significantly from a theoretic point of view (see Section~\ref{section:discussion}), they are algorithmically closely related. For details on how to adapt the algorithm for non-binary trees, we refer to the techniques used in~\cite{DBLP:journals/algorithmica/Zhang96} and~\cite{wetzels2022branch}.

Again, the algorithm is based on identifying subtrees through pairs of nodes. Subtrees rooted in an edge as well as subtrees that are created through subtraction are identified by their root and its unique child. E.g.\ for a binary abstract merge tree $T$ with $\troot(T)=v_1$, $(v_2,v_1) \in E(T)$ and $(v_3,v_2),(v_4,v_2) \in E(T)$, we identify $T[(v_2,v_3)]$ by $v_2,v_3$ and $T - T_1[(v_2,v_3)]$ by $v_1,v_4$. The recursion in Lemma~\ref{lemma:rec_inner} can be adapted to this notation, which is illustrated in Figure~\ref{fig:recursion_dp} for four of the six recursive cases. By returning the minimum of the six results for inner nodes and adding base cases, we obtain Algorithm~\ref{alg:odedist}, which computes the here defined one-degree edit distance for abstract merge trees.

Now consider the running time of Algorithm~\ref{alg:odedist}. For two binary abstract merge trees $T_1,T_2$, there are at most $|T_1|^2 \cdot |T_2|^2$ pairs of paths or 4-tuples of vertices. Since the number of subproblems for each pair is constant, the running time has an upper bound of $|T_1|^2 \cdot |T_2|^2$ when using memoization. For trees of arbitrary degree, we would get another factor of $\mathcal{O}((\deg_1+\deg_2) \cdot \log(\deg_1+\deg_2))$, similar to the algorithms in\cite{DBLP:journals/algorithmica/Zhang96} or\cite{wetzels2022branch}.

\begin{figure*}[]
\centering
\includegraphics[width=0.25\linewidth]{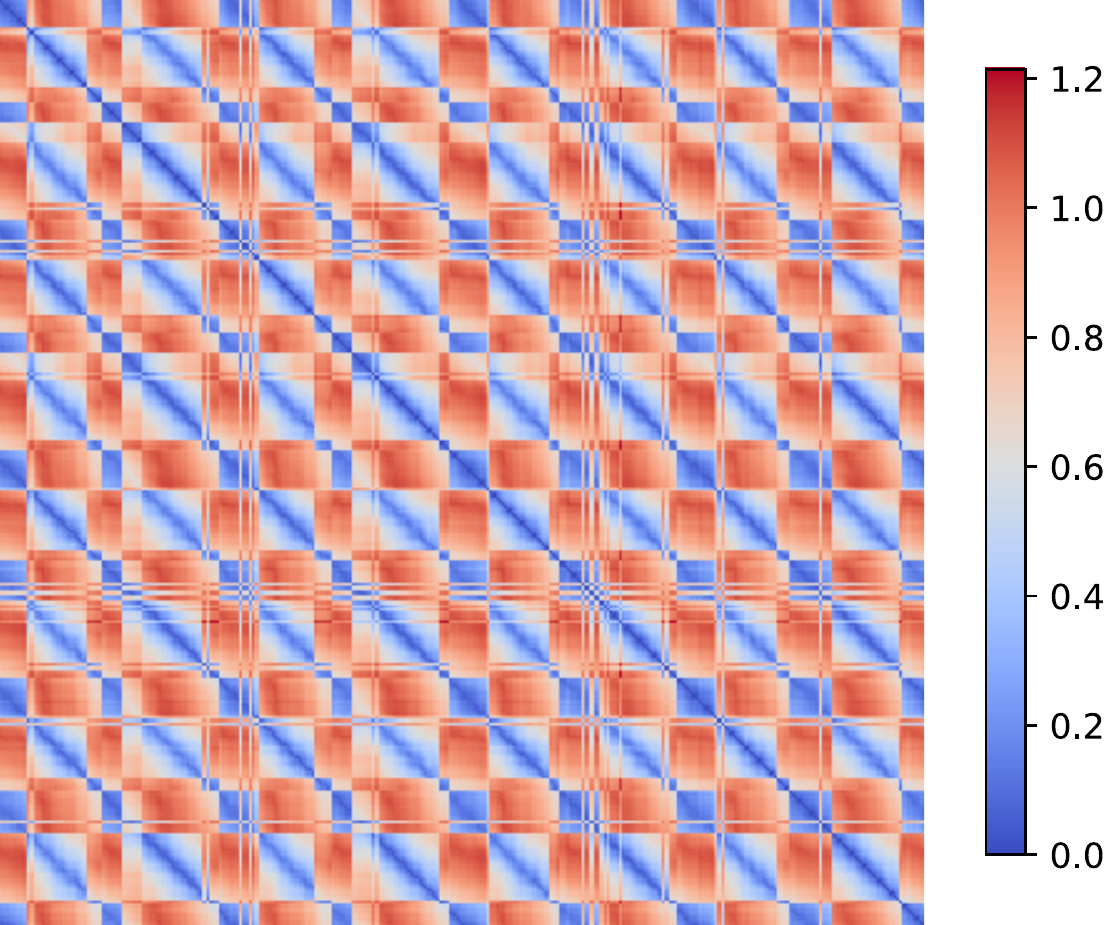}
\hspace{0.05\linewidth}
\includegraphics[width=0.25\linewidth]{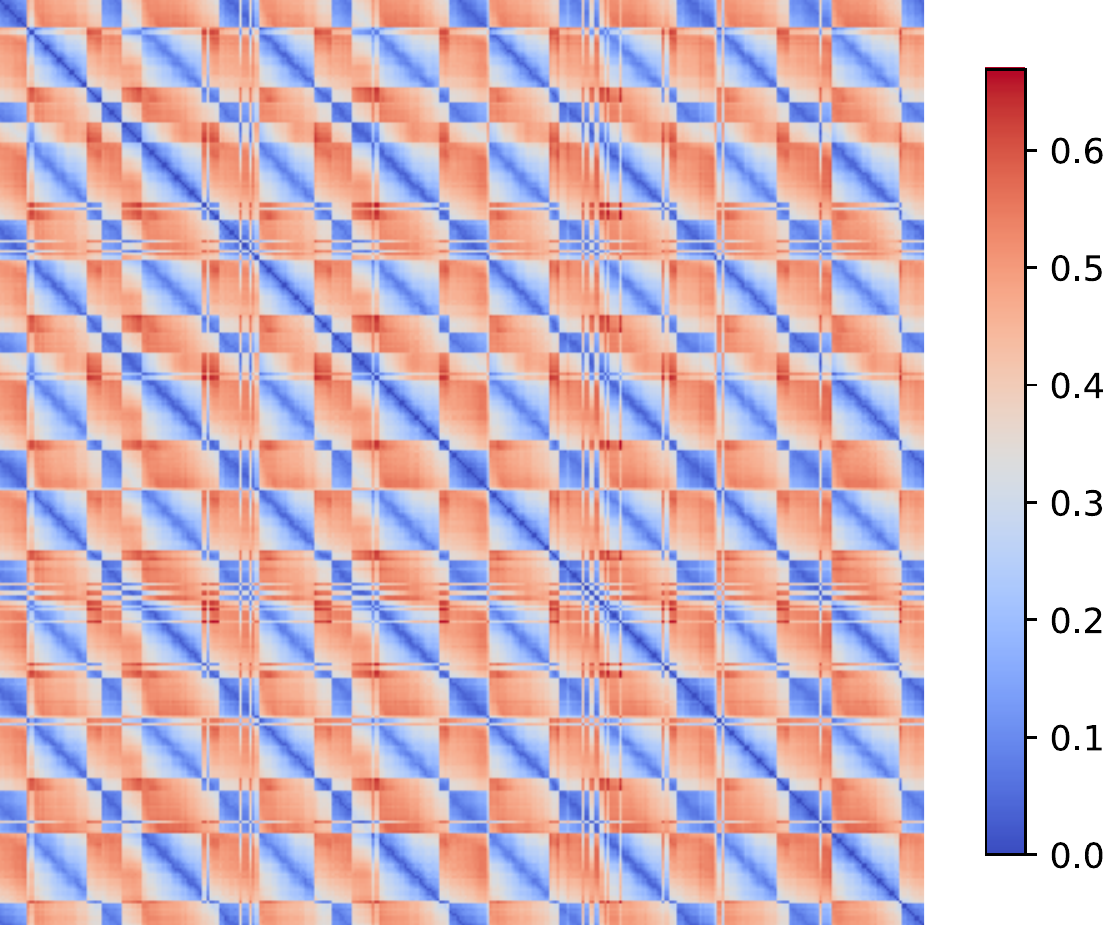}
\hspace{0.05\linewidth}
\includegraphics[width=0.25\linewidth]{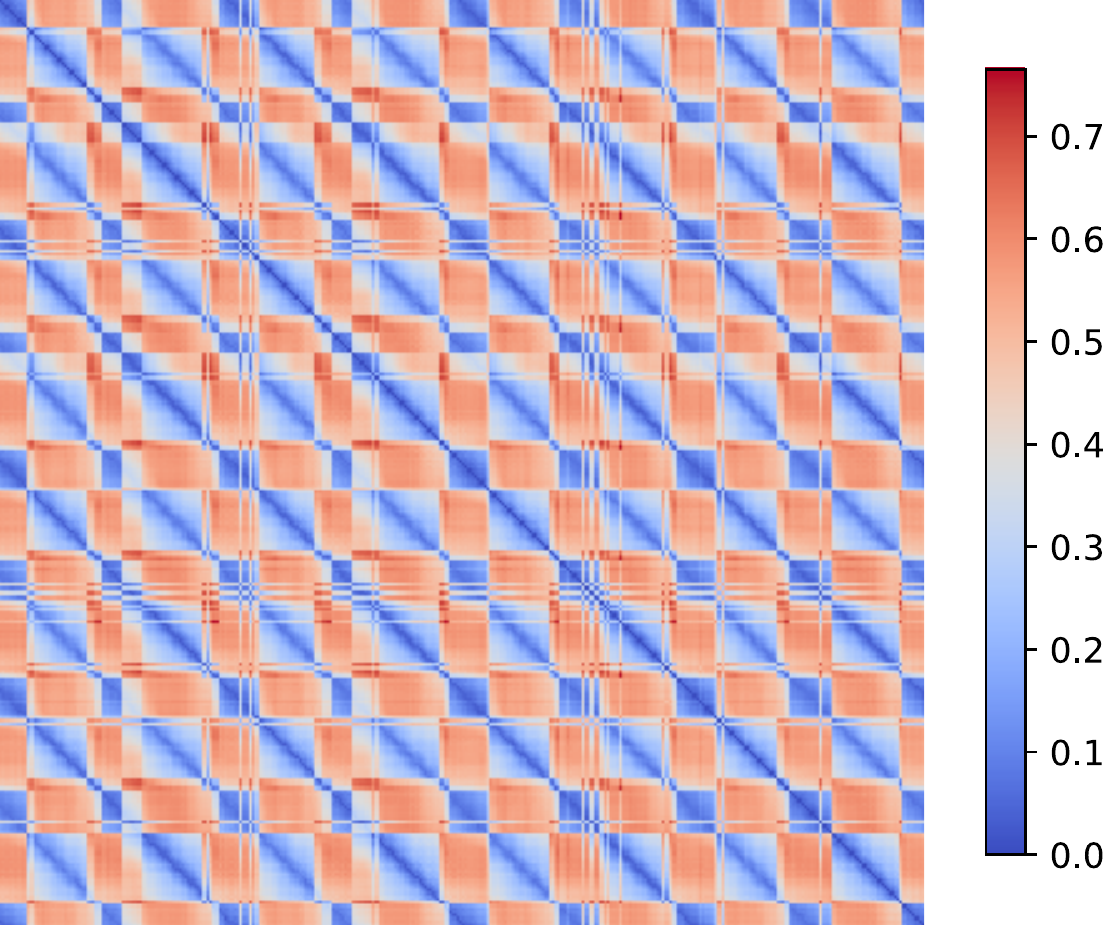}
\caption{The distance matrices of the vortex street dataset using the path mapping distance (left), the branch mapping distance (middle) and the classic constrained edit distance (right), all visualized through heatmaps.}
\label{fig:dm_weinkauf_path}
\end{figure*}

\subsection*{Experiments}

In the following, we demonstrate the utility of our technique as a basis for typical tasks in visualization. We apply the new distance to two datasets that were also used in~\cite{wetzels2022branch}. The basis for these experiments is a C++ implementation of Algorithm~\ref{alg:odedist}. The merge trees were computed using TTK~\cite{DBLP:journals/tvcg/TiernyFLGM18}. Our implementation is publicly available on Github~\cite{repository}. Computation times of single distances for the here used simplified trees were of the same order of magnitude as the closely related branch mapping distance and follow the theoretical bounds. However, the path mapping distance performed slightly better with a speed-up factor of 1.7 on average. Table~\ref{tab:runtimes} shows the comparison in more detail. This speedup is due to the simplified branching of the recursion, specifically in the case of deletions of whole subtrees, since the path mapping distance does not have to try all branch decompositions in this case and can just add up all edge persistences.

\begin{table}[b]
\centering
\scalebox{0.8}{
\begin{tabular}{l|ccccccccc}
     & HC (10) & C (18) & O (20) & VS (68) & HC (233)  \\ \hline
  $d_B$ & $17.5\cdot10^{-6}$s & $7.0\cdot10^{-5}$s & $10.4\cdot10^{-5}$s & $3.5\cdot10^{-3}$s & $6.6$s  \\
  $d_P$ & $9.6\cdot10^{-6}$s & $4.0\cdot10^{-5}$s & $6.3\cdot10^{-5}$s & $2.0\cdot10^{-3}$s & $4.0$s  \\
\end{tabular}
}
\caption{Running times of the branch mapping distance and path mapping distance on datasets from~\cite{wetzels2022branch}: The synthetic outlier ensemble (O), the cluster example (C), the heated cylinder (HC) and the vortex street (VS). The sizes of the merge trees are shown in brackets. All times were obtained on a standard workstation with an Intel Core i7-7700 and 64GB of RAM.}
\label{tab:runtimes}
\end{table}

The first dataset on which we apply our new distance is the outlier ensemble from~\cite{wetzels2022branch}. It consists of 20 scalar fields with merge trees of 20 nodes. It demonstrates the branch decomposition-independence of a distance measure if no clusters are found except a single outlier. A more detailed discussion on this behavior can be found in~\cite{wetzels2022branch}. Figure~\ref{fig:dm_outlier2_path} shows that the path mapping distance performs in the expected way and yields very similar results to the branch mapping distance. Furthermore, it can be seen that the results of the two branch decomposition-independent distances differ significantly from the results using fixed branch decompositions, since they do not show false clusters and are therefore able to identify the outlier clearly.

The second dataset is a time-varying scalar field consisting of 1001 time steps representing the velocity magnitude of the flow around a cylinder that forms a periodic Kármán vortex street. It was simulated by Weinkauf~\cite{weinkauf10c} using the \emph{GerrisFlowSolver}~\cite{gerrisflowsolver}. Figure~\ref{fig:dm_weinkauf_path} shows distance matrices of the timeline using different edit distances. The periodic pattern is clearly visible, with the same periods identified by the path mapping distance, the branch mapping distance from~\cite{wetzels2022branch} and the constrained edit distance on BDTs from~\cite{DBLP:journals/tvcg/SridharamurthyM20}.

\section{Discussion}
\label{section:discussion}

In this section, we discuss how the here introduced edit distance and path mappings compare to other edit distances, especially branch mappings, and also elaborate further on the choice of edit operations.

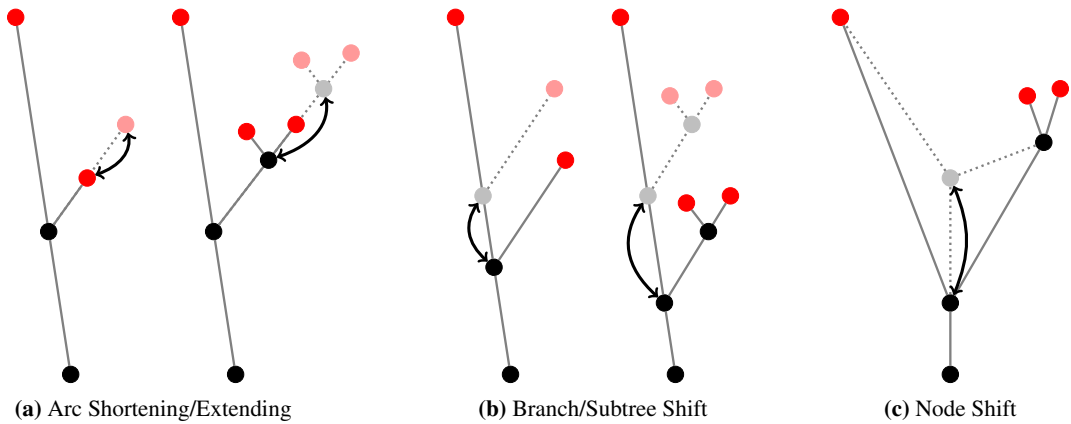
\begin{figure*}[]
 \centering 
 \resizebox{0.8\linewidth}{!}{
 \begin{tikzpicture}[yscale=0.65]
 \definecolor{lightred}{rgb}{1,0.6,0.6}
 
 \node[draw,circle,fill] at (0, 0) (root) {};
 \node[draw,circle,fill] at (-0.4, 4) (s1) {};
 \node[draw,circle,fill,red] at (-1, 10) (m1) {};
 
 \node[draw,circle,fill,lightred] at (1, 7) (m2') {};
 \draw[gray,very thick,dotted] (s1) -- (m2');
 
 \node[draw,circle,fill,red] at (0.3, 5.5) (m2) {};
 \draw[gray,very thick] (root) -- (s1);
 \draw[gray,very thick] (s1) -- (m1);
 \draw[gray,very thick] (s1) -- (m2);
 
 \draw[ultra thick,<->] (m2') to[bend left=35] (m2);
 
 
 \node[draw,circle,fill] at (3+0, 0) (root_) {};
 \node[draw,circle,fill] at (3-0.4, 4) (s1_) {};
 
 \node[draw,circle,fill,lightred] at (3+2.1, 9) (m2'_) {};
 \draw[gray,very thick,dotted] (s1_) -- (m2'_);
 \node[draw,circle,fill,lightgray] at (3+1.6, 8) (s2'_) {};
 \node[draw,circle,fill,lightred] at (3+1.2, 8.8) (m3'_) {};
 \draw[gray,very thick,dotted] (s2'_) -- (m3'_);
 
 \node[draw,circle,fill,red] at (3-1, 10) (m1_) {};
 \node[draw,circle,fill,red] at (3+1.1, 7) (m2_) {};
 \draw[gray,very thick] (root_) -- (s1_);
 \draw[gray,very thick] (s1_) -- (m1_);
 \draw[gray,very thick] (s1_) -- (m2_);
 \node[draw,circle,fill] at (3+0.6, 6) (s2_) {};
 \node[draw,circle,fill,red] at (3+0.2, 6.8) (m3_) {};
 \draw[gray,very thick] (s2_) -- (m3_);
 
 \draw[ultra thick,<->] (s2_) to[bend right=35] (s2'_);
 
 
 \node[draw,circle,fill] at (8+0, 0) (root) {};
 \node[draw,circle,fill] at (8-0.3, 3) (s1) {};
 \node[draw,circle,fill,red] at (8-1, 10) (m1) {};
 \node[draw,circle,fill,red] at (8+1, 6) (m2) {};
 \draw[gray,very thick] (root) -- (s1);
 \draw[gray,very thick] (s1) -- (m1);
 \draw[gray,very thick] (s1) -- (m2);
 
 \node[draw,circle,fill,lightgray] at (8-0.5, 5) (s1') {};
 \node[draw,circle,fill,lightred] at (8+0.8, 8) (m2') {};
 \draw[gray,very thick,dotted] (s1') -- (m2');
 
 \draw[ultra thick,<->] (s1) to[bend left=30] (s1');
 
 
 \node[draw,circle,fill] at (0+11, 0) (root_) {};
 \node[draw,circle,fill] at (-0.2+11, 2) (s1_) {};
 \node[draw,circle,fill,red] at (-1+11, 10) (m1_) {};
 \node[draw,circle,fill,red] at (1+11, 5) (m2_) {};
 \draw[gray,very thick] (root_) -- (s1_);
 \draw[gray,very thick] (s1_) -- (m1_);
 \draw[gray,very thick] (s1_) -- (m2_);
 \node[draw,circle,fill] at (0.6+11, 4) (s2_) {};
 \node[draw,circle,fill,red] at (0.2+11, 4.8) (m3_) {};
 \draw[gray,very thick] (s2_) -- (m3_);
 
 \node[draw,circle,fill,lightgray] at (-0.5+11, 5) (s1'_) {};
 \node[draw,circle,fill,lightred] at (0.7+11, 8) (m2'_) {};
 \draw[gray,very thick,dotted] (s1'_) -- (m2'_);
 \node[draw,circle,fill,lightgray] at (0.3+11, 7) (s2'_) {};
 \node[draw,circle,fill,lightred] at (-0.1+11, 7.8) (m3'_) {};
 \draw[gray,very thick,dotted] (s2'_) -- (m3'_);
 
 \draw[ultra thick,<->] (s1_) to[bend left=30] (s1'_);
 
 
 \node[draw,circle,fill] at (0+16, 0) (root__) {};
 \node[draw,circle,fill] at (0+16, 2) (s1__) {};
 \node[draw,circle,fill,red] at (-2+16, 10) (m1__) {};
 \node[draw,circle,fill,red] at (2+16, 8) (m2__) {};
 \draw[gray,very thick] (root__) -- (s1__);
 \draw[gray,very thick] (s1__) -- (m1__);
 \node[draw,circle,fill] at (1.7+16, 6.5) (s2__) {};
 \node[draw,circle,fill,red] at (1.4+16, 7.8) (m3__) {};
 \draw[gray,very thick] (s1__) -- (s2__);
 \draw[gray,very thick] (s2__) -- (m2__);
 \draw[gray,very thick] (s2__) -- (m3__);
 
 \node[draw,circle,fill,lightgray] at (0+16, 5.5) (s1'__) {};
 \draw[gray,very thick,dotted] (s1'__) -- (m1__);
 \draw[gray,very thick,dotted] (s1__) -- (s1'__);
 \draw[gray,very thick,dotted] (s1'__) -- (s2__);
 
 \draw[ultra thick,<->] (s1__) to[bend right=15] (s1'__);
 
 
 \node[] at (1.5, -1) (label1) {\Large\textbf{(a)} Arc Shortening/Extending};
 \node[] at (9.5, -1) (label1) {\Large\textbf{(b)} Branch/Subtree Shift};
 \node[] at (16, -1) (label1) {\Large\textbf{(c)} Node Shift};

 \end{tikzpicture}
 }
\caption{The three types of possible edit operations on merge trees.}
\label{fig:edit_operations_continuous}
\end{figure*}

\subsection*{Edit Operations}

We begin by listing alternative edit operations on merge trees. We identified the following operations to be considerable: stretching and shrinking of arcs or branches, shifting of branches or subtrees, and shifting of nodes. Figure~\ref{fig:edit_operations_continuous} shows examples for these three classes. For identifying the operations, we used the following core assumption: the edit operations and their costs should be rooted the scalar function. Typical tools like persistence, Wasserstein metrics or similar concepts, which are used in other distance measures, do exactly this. Note that the here considered edit operations are all variants of relabel operations. Deletions and insertions should behave roughly the same in all models.

Now we first consider node shifts. Although these are based on the scalar function, they do actually not represent the changes that we are interested in: we are usually not interested in the \emph{absolute} scalar values, but rather the \emph{relative} ones, i.e.\ for the topological similarity of two merge trees, we only want to consider the distance to the root or the length/persistence of features. Hence, node shifts are not the operations we want, at least if we use scalar difference as the cost measure. Not using these operations also does not restrict the expressiveness of the edit distance, i.e.\ all node shifts can also be expressed as a sequence of branch/subtree shifts or a sequence of stretch and shrink operations.

Next, we consider branch shifts and subtree shifts. They can be expressed through branch based edit mappings and are, in fact, closely related to branch based edit distances like the Wasserstein distance for merge trees~\cite{DBLP:journals/tvcg/PontVDT22} or the branch mapping distance~\cite{wetzels2022branch}. It is actually possible to define a base metric such that the branch mapping distance exactly represents an edit distance based on these operations. However, this edit distance differs from typical ones in various ways. 

First, we have to restrict the sequences of edit operation to those that only touch a node once. Since a node in a merge tree can belong to multiple branches, it can be modified multiple times through edit operations using different branches. This leads to the problem that these sequences would no longer correspond to the mappings, hence, we need to restrict them. Second, in contrast to classic edit distances, the branch based operations are not local ones. Typical edit distances modify only one node or edge locally, whereas branch based operations modify a complete branch and, depending on the definition, also the descending branches. One could argue that those two problems are only aesthetic ones, but they are actually the core of a third problem, which leads to disadvantages in practice: an edit distance based on branch shifts is not a metric (if the branch decomposition-independent variant is chosen) or depending on a fixed BDT (see~\cite{wetzels2022branch} for a detailed discussion on this problem). Intuitively, the reason for this is that through the use of different branch decompositions, it can actually be cheaper to go over an intermediate tree than using the optimal branch mapping between the original and resulting tree, which contradicts the triangle inequality. A formal proof can be found in~\cite{wetzels2022branch}, where they actually use the shifting base metric. To conclude this argument, for branch shifts we have to chose between the metric property and losing the correspondence to mappings (which also means efficient computability).

The remaining operations are shrinking and stretching. Due to the strong correspondence to classic edit operations on trees, we chose stretching and shrinking of \emph{edges} to be the natural model. They are local operations, lead to a metric distance, and can transform any merge tree into any other merge tree, i.e.\ they can also express all other operations. Furthermore, they naturally correspond to deformation retractions and their inverse deformations.

\subsection*{Comparison}

We now compare the here introduced edit distance and path mappings to previous methods. As mentioned in the discussion on edit operations, branch mappings capture, in essence, optimal mappings achieved through shifts of branches in a merge tree. Hence, other methods that are based on fixed BDTs (e.g.\ those from~\cite{DBLP:journals/tvcg/SridharamurthyM20,DBLP:journals/tvcg/PontVDT22,DBLP:journals/cgf/SaikiaSW14}) do the same, but only allow shifts of certain branches from a fixed decomposition. In contrast to that, path mappings capture edge based operations like stretching or shrinking. The two operation sets differ significantly from a theoretic point of view which shows in the fact that one leads to a metric while the other does not. Furthermore, the intuitions behind the operations are also completely different, as one of them is a local operation while the other one operates on global structures in a merge tree. This is an interesting observation considering the fact that algorithmically, i.e.\ in their recursive structure, path mappings and branch mappings are almost identical. To sum up, path mappings show the following behavior:
\begin{itemize}
    \item In contrast to classic edit operations, path mappings do not suffer from the problems shown in Figure~\ref{fig:classicVsContinuousEdit}, i.e.\ they are \emph{well-defined on merge trees}.
    \item In contrast to branch mappings, path mappings lead to a \emph{metric distance function}.
    \item In contrast to typical edit distances on BDTs, path mappings are \emph{branch decomposition-independent}.
    \item They are \emph{less efficient} to compute than classic edit distances (either on merge trees or BDTs).
\end{itemize}
Based on these properties, path mappings can be seen as an alternative or an improvement for branch mappings that has the same advantages over classic edit distances, but also the same increased complexity. Although the metric property did not yield improved results over the branch mapping distance in our experiments, it is an important property in practice since it allows to use the distance in more advanced analysis methods. For example, due to the metric property, path mappings could be used to compute geodesics and barycenters of merge trees using similar techniques to those from~\cite{DBLP:journals/tvcg/PontVDT22}, which we believe to not be possible or at least much harder with branch mappings. However, we leave this integration for future work. Furthermore, we should note that due to the nature of constrained edit distances, the path mapping distance is susceptible to the same saddle-swap instabilities as other merge tree edit distances in~\cite{wetzels2022branch,DBLP:journals/tvcg/SridharamurthyM20,DBLP:journals/cgf/SaikiaSW14,DBLP:journals/tvcg/PontVDT22}, but it is possible to apply the typical preprocessing to reduce this problem.

\section{Conclusion}
\label{section:conclusion}

In this paper, we defined a new edit distance for merge trees based on geometric operations on the continuous object, that resembles an intuitive adaptation of classic tree edit distances to merge trees much closer than branch based methods. We summarized its advantages and limitations in Section~\ref{section:discussion} and presented a short demonstration of its utility in practice in Section~\ref{section:experiments}. We also provide an open-source implementation publicly available on GitHub.

In future work, we want to study stability properties of the new distance (specifically comparing the unconstrained and one-degree versions in this regard), parallel algorithms for more practical running times on complex datasets, and a possible adaptation to contour trees. Furthermore, the path mapping distance could be integrated in more advanced edit distance-based visualization techniques such as barycenter merge trees~\cite{DBLP:journals/tvcg/PontVDT22} or alignments~\cite{DBLP:journals/cgf/LohfinkWLWG20}.

\acknowledgments{
The authors wish to thank Markus Anders, Heike Leitte and Jonas Lukasczyk for their valuable input. This research was funded by the Deutsche Forschungsgemeinschaft (DFG, German Research Foundation) – 442077441.}

\bibliographystyle{abbrv-doi}

\bibliography{main}
\end{document}




\maketitle

\appendix

\section*{Appendix A: Proof of Theorem 1}
\label{section:proof1}

We now provide a proof of Theorem~1. For better readability, we also state the theorem again.

\begin{theorem}
Let $\mathcal{T}_1,f_1$ and $\mathcal{T}_2,f_2$ be merge trees with abstractions $T_1=T(\mathcal{T}_1),\ell_{f_1}$ and $T_2=T(\mathcal{T}_2),\ell_{f_2}$. If $|\mathcal{T}_2|$ is homeomorphic to a deformation retract of $|\mathcal{T}_1|$, then there is a sequence of edit operations $S$ only containing deletions and relabels that decrease the edge labels such that $T_1,\ell_{f_1} \xrightarrow{\scriptscriptstyle S} T_2,\ell_{f_2}$. \textbf{(A)}

Furthermore, given two abstract merge trees $T_1,\ell_1$ and $T_2,\ell_2$ with $T_1,\ell_1 \xrightarrow{\scriptscriptstyle S} T_2,\ell_2$ and $S$ only containing deletions and relabels that decrease the edge labels, then there are merge trees $\mathcal{T}_1,f_1$ and $\mathcal{T}_2,f_2$ with $T_1=T(\mathcal{T}_1),\ell_1=\ell_{f_1}$ and $T_2=T(\mathcal{T}_2),\ell_2=\ell_{f_2}$, such that $|\mathcal{T}_2|$ is homeomorphic to a deformation retract of $|\mathcal{T}_1|$. \textbf{(B)}
\end{theorem}

\begin{proof}
\textbf{(A)} Assume that $|\mathcal{T}_2|$ is homeomorphic to a deformation retract $\mathcal{T}'_1$ of $|\mathcal{T}_1|$, i.e.\ $|\mathcal{T}'_1| \subseteq |\mathcal{T}_1|$ and there is a continuous function $r: |\mathcal{T}_1| \times [0,1] \rightarrow |\mathcal{T}_1|$ with $r(x,t) = x$ for all $x \in |\mathcal{T}'_1|,\ t \in [0,1]$ and $r(x,0) = x$ for all $x \in |\mathcal{T}_1|$. Now consider $r(|\mathcal{T}_1|,t)$ for some $t \in [0,1]$. It is a merge tree, i.e.\ homeomorphic to the underlying space of another simplicial complex $K_t$ with $\text{Vert}(K_t)$ being the critical points of $|K_t|$. It should be easy to see that w.l.o.g.\ $K_0 = \mathcal{T}_1$ and $K_1 = \mathcal{T}_2$. If $|\mathcal{T}_2|$ is homeomorphic to $|\mathcal{T}_1|$, then so are all $r(|\mathcal{T}_1|,t)$. Otherwise, there is a finite number of time points $0=t_1<t_2<...<t_k$ such that $r(|\mathcal{T}_1|,t)$ is homeomorphic to $r(|\mathcal{T}_1|,t')$ but not to $r(|\mathcal{T}_1|,t'')$ for all $t_i \leq t \leq t' < t_{i+1} \leq t''$, $0 \leq i < k$. To see why, note that each non-homeomorphic subspace of $|\mathcal{T}_1|$ has to remove at least one 1-simplex from $\mathcal{T}_1$, of which there are only finitely many. With the same argument, we can follow that at each $t_i$ a finite set of 1-simplices and their faces are removed from $K_{t_i-\epsilon}$ to obtain $K_{t_i}$. Hence, we can follow that $T(K_{t_{i-1}}) \xrightarrow{\scriptscriptstyle S} T'$ where $T'$ is isomorphic to $T(K_{t_i})$ and $S$ only consists of deletions. By relabeling the edges in $T'$ where necessary, we obtain a sequence that transforms $T(K_{t_{i-1}})$ into $T(K_{t_i})$. We can concatenate these sequences to obtain one from $T(K_0)$ to $T(K_1)$.

\textbf{(B)} Now assume that $T_1,\ell_1 \xrightarrow{\scriptscriptstyle S} T_2,\ell_2$ and $S$ only contains deletions and relabels that decrease the edge labels. By unfolding $S = o_1...o_k$ we get
$ (T_1,\ell_1) = (T'_0,\ell'_0) \xrightarrow{\scriptscriptstyle o_1} (T'_1,\ell'_1) \xrightarrow{\scriptscriptstyle o_2} ... \xrightarrow{\scriptscriptstyle o_k} (T'_k,\ell'_k) = (T_2,\ell_2) $,
where each $o_i$ is an edge contraction or a relabel. We can now construct a continuous merge tree $\mathcal{T}_i,f_i$ and a deformation retraction $r_i$ for each $T'_i$. Let $\mathcal{T}_0$ be a simplicial complex representing some embedding of $T'_0$ and $f_0$ be the function defined as follows: the vertex of $\mathcal{T}_0$ corresponding to $\troot(T'_0)$ gets function value $0$, all other vertices the scalar distance to the root defined through the edge labels and the scalars of all other points in the underlying space are defined through linear interpolation. If $o_1$ is a deletion of an edge $(c,p)$, we define $r_1$ to be the deformation retraction that contracts the 1-simplex corresponding to $(c,p)$ to the point corresponding to $p$. If $o_1$ is a relabeling of $(c,p)$, we just contract that part of the corresponding arc in $\mathcal{T}_0$ that has scalar distance to $p$ larger than $\ell'_0((c,p))$. $\mathcal{T}_1$ is then just the image of $r_0$ and $f_1$ is defined through $\ell'_1$ as before. We can do so because $\mathcal{T}_1$ is homeomorphic to a simplicial complex $K$ isomorphic to $T'_1$ (a 1-cell is removed from $\mathcal{T}_0$ if and only if an edge is removed from $T'_0$). Hence, $T(\mathcal{T}_1,f_1) = T'_1,\ell'_1$. Repeating this for each $o_i$ and concatenating the $r_i$, we get a deformation retraction from $\mathcal{T}_0$ to $\mathcal{T}_k$, which is again homeomorphic to some actual merge tree $\mathcal{T}'_k$ isomorphic to $T'_k = T_2$.
\end{proof}

\section*{Appendix B: Proof of Lemma 2}
\label{section:proof2}

Next, we prove Lemma~2. again, we repeat the lemma for completeness.

\newtheorem*{lemma2}{Lemma 2}

\begin{lemma2}

Let $M \subseteq \mathcal{P}(T_1) \times \mathcal{P}(T_2)$ be an optimal path mapping between two abstract merge trees $T_1,T_2$. For any path $p \in \mathcal{P}(T_1)$ that is contained in $M$, there are at least two paths $p',p''$ that are contained in $M$ with $\pathstart(p') = \pathstart(p'') = \pathend(p)$.

Symmetrically, the same holds for any path $p \in \mathcal{P}(T_2)$.

\end{lemma2}

\begin{proof}
Assume that $M$ contains two paths $p_1 = v_1...v_{k_1} \in \mathcal{P}(T_1)$ and $p_2 = u_1...u_{k_2} \in \mathcal{P}(T_1)$ with $u_1 = v_{k_1}$, and for all $c \neq u_2$ with $(c,v_{k_1}) \in E(T_1)$, the subtree $T_1[(c,v_{k_1})]$ is not contained in $M$. Let $q_1 = v'_1...v'_{k'_1}$ and $q_2 = u'_1...u'_{k'_2}$ be the mapped paths, i.e.\ $(p_1,q_1),(p_2,q_2) \in M$. Then we can construct a path mapping $M'$ with $\cost(M') \leq \cost(M)$ as follows.

Let $p = v_1...v_{k_1}u_2...u_{k_2}$ and $q = v'_1...v'_{k'_1}u'_2...u'_{k'_2}$, then we have $\ell_1(p) = \ell_1(p_1) + \ell_1(p_2)$ and $\ell_2(q) = \ell_2(q_1) + \ell_2(q_2)$, and thus $\cost(\ell_1(p),\ell_2(q)) \leq \cost(\ell_1(p_1),\ell_2(q_1)) + \cost(\ell_1(p_2),\ell_2(q_2))$. We define $M'$ as
$$(M \setminus \{(p_1,q_1),(p_2,q_2)\}) \cup \{(p,q)\}.$$
Since $\cost(M') = \cost(M) - (\cost(\ell_1(p_1),\ell_2(q_1)) + \cost(\ell_1(p_2),\ell_2(q_2))) + \cost(\ell_1(p),\ell_2(q))$, we know that $\cost(M') \leq \cost(M)$.

It remains to check whether $M'$ is a valid path mapping. The one-to-one condition (1) is of course not changed. For conditions (2) and (3), it is enough to check the condition from Lemma~1. Since it was fulfilled for all inner nodes of $p_1$ and $p_2$ ($M$ is a valid path mapping), we just have to check $v_{k_1}$, for which holds by assumption.

Therefore, $M'$ is an optimal path mapping as well and because of this we can assume $M$ to not contain any two such paths $p_1,p_2$, which means it fulfills the desired property.
\end{proof}

\section*{Appendix C: Proof of Lemma 3}
\label{section:proof3}

Next, we prove Lemma~3. again, we repeat the lemma for completeness.

\newtheorem*{lemma3}{Lemma 3}

\begin{lemma3}

Let $S$ be a cost-optimal sequence of edit operations that transforms an abstract merge tree $T_1$ into another one $T_2$. Then there exists a path mapping $M \subseteq \mathcal{P}(T_1) \times \mathcal{P}(T_2)$ such that $\cost(M) \leq \cost(S)$.

\end{lemma3}

\begin{proof}
Since $S$ is optimal, we can assume that it does not contain any redundant operations, i.e.\ no edge is relabeled twice or relabeled and inserted/deleted, and an edge contraction and its exact inverse do not appear in $S$ together.

Furthermore, we can reorder $S$ in the following way without changing the resulting tree and the cost:
\begin{itemize}
\item Relabel operations can be placed in arbitrary positions, since they do not interfere with any other operation. If we swap them with an insert or delete operation, the edge to relabel might change, but the result stays the same.
\item Insert and delete operations can also be swapped freely, since they only appear on edges to leaves. Note that only pairs of one insert and one delete operation can be swapped, as there can be pairs of two insertions or pairs of two deletions that depend on each other, e.g.\ an insertion of a new leaf node and edge followed by an insertion of another leaf and corresponding saddle that divides the before inserted edge.
\end{itemize}
Hence, we can assume $S$ to have the following structure: first, a sequence $S_1$ of deletions, then a sequence $S_2$ of relabels, followed by a sequence $S_3$ of insertions. Now we need to construct the path mapping $M$ from these three sequences. Let $T'$ be the tree obtained from $T_1$ by $S_1$ and $T''$ be the one obtained from $T'$ by $S_2$.

For $S_1$, we can construct a path mapping $M_1$ between $T_1$ and $T'$ with $\cost(M_1) \leq \cost(S_1)$ as follows: $M_1 = \{ (\tpath_{T_1}(c,p),(c,p)) \mid (c,p) \in E(T') \}$. To see the equivalence, consider how $S_1$ changes $T_1$. Each operation only removes one edge and, in some cases, merges two other edges into one. Thus, for each operation, we get a surjective mapping from the paths of the first tree to the edges of the second tree, which can be concatenated to obtain $M_1$. We omit a formal proof by induction.

Since $S_2$ does not change the structure of $T'$, we get a bijective path mapping $M_2$ between the edges of $T'$ and the edges $T''$ with $\cost(M_2) \leq \cost(S_2)$. Furthermore, since it is a bijection between the edges of the two trees, we can conclude that $M_2 \circ M_1$ is a path mapping between $T_1$ and $T''$ and $\cost(M_2 \circ M_1) \leq \cost(S_1S_2)$, since euclidean distance is a metric on $\mathbb{R}$.

For $S_3$, we can reverse the process for $S_1$, i.e.\ we can construct a path mapping $M_3$ between $T''$ and $T_2$ with $\cost(M_3) \leq \cost(S_3)$ through $M_3 = \{ ((c,p),\tpath_{T_2}(c,p)) \mid (c,p) \in E(T'') \}$. Again, we omit a formal proof by induction.

We can then concatenate $M_1,M_2,M_3$ like $S_1,S_2,S_3$ to obtain $M=M_3 \circ M_2 \circ M_1$ with $\cost(M) \leq \cost(S)$. To see why, consider the image of $M_2 \circ M_1$ and the pre-image of $M_3$. Both are exactly $E(T'')$, hence, we can concatenate the two mappings and for each pair, the distance can only get less, since euclidean distance is a metric on $\mathbb{R}$.
\end{proof}

\section*{Appendix D: Proof of Lemma 4}
\label{section:proof4}

Next, we prove Lemma~4. again, we repeat the lemma for completeness.

\newtheorem*{lemma4}{Lemma 4}

\begin{lemma4}

For two abstract merge trees $T_1,T_2$, let $M \subseteq \mathcal{P}(T_1) \times \mathcal{P}(T_2)$ be a path mapping. Then there exists a sequence $S$ of edit operations that transforms $T_1$ into $T_2$, with $\cost(S) = \cost(M)$.

\end{lemma4}

\begin{proof}
Let $T_1^M$ be the tree that consists of the union of all paths from $T_1$ that are contained in $M$. By Lemma~2, we can conclude that $T_1^M$ is an abstract merge tree. We define $T_2^M$ symmetrically. Note that $T_1^M$ and $T_2^M$ are structurally isomorphic, i.e.\ when ignoring the labels.

Furthermore, let $D_1 \subseteq E(T_1)$ be the set of edges in $T_1$ that are not contained in $M$, and $D_2 \subseteq E(T_2)$ those not contained in $M$ from $T_2$. From condition~3 in Definition~2 we can conclude that $D_1$ and $D_2$ can be partitioned into $D_1^1,...,D_1^{k_1}$ and $D_2^1,...,D_2^{k_2}$ such that each $D_i^j$ forms a complete subtree of $T_1$/$T_2$, since any path contained in $M$ has to be connected to the root through other contained paths.

For each of the $D_1^j$, we therefore can construct a sequence $S_1^j$ of edge deletions that removes the complete corresponding subtree. The cost $\cost(S_1^j)$ of this sequence is the sum of the labels of all edges in $D_1^j$. Note that deleting a merged edge has exactly the same cost as deleting the merged individually, since the merged label is defined through the sum. With this observation, we get that
$$ \sum_{(l_1,l_2) \in \mappingdelete(M)}\cost(l_1,l_2) = \sum_{1 \leq j \leq k_1} \cost(S_1^j). $$
Furthermore, it should be easy to see that $S_1^1...S_1{k_1}$ transforms $T_1$ into $T_1^M$.
Symmetrically, we can construct a sequence of insertions $\cost(S_2^j)$ for each $D_2^j$, such that $S_2^1...S_2{k_2}$ transforms $T_2^M$ into $T_2$ and
$$ \sum_{(l_1,l_2) \in \mappinginsert(M)}\cost(l_1,l_2) = \sum_{1 \leq j \leq k_2} \cost(S_2^j). $$

As a last step, we now define a sequence $S_M$ of relabel operations that transforms $T_1^M$ into $T_2^M$ such that
$$ \sum_{(l_1,l_2) \in \mappingrelabel(M)}\cost(l_1,l_2) = \cost(S_M). $$
Then, we can conclude that $S = S_1^1...S_1{k_1}S_MS_2^1...S_2{k_2}$ transforms $T_1$ into $T_2$ and has exactly the same costs as $M$.

To construct $S_M$, note that for each pair $(p_1,p_2) \in M$, there are edges $e_1$ in $T_1^M$ with $\ell_1^M(e_1) = \ell_1(p_1)$ and $e_2$ in $T_2^M$ with $\ell_2^M(e_2) = \ell_2(p_2)$, and vice versa. Hence, we can construct $S_M$ by taking the corresponding edge relabels for each pair of paths in $M$.
\end{proof}

\bibliographystyle{abbrv-doi}
